\documentclass[letterpaper,twocolumn,10pt]{article}
\usepackage{usenix2019_v3}

\usepackage{authblk}
\usepackage{tikz}
\usepackage{amsmath,amssymb,amsthm}

\usepackage{cite}
\usepackage{hyperref}
\usepackage{subcaption}

\usepackage{xcolor}

\usepackage{graphicx}

\usepackage[normalem]{ulem}
\usepackage{siunitx}
\usepackage{mathpartir}

\newcommand{\sys}{Pacer}
\newcommand{\eg}{e.g.,~}
\newcommand{\ie}{\emph{i.e.},~}
\newcommand{\etal}{\emph{et al.}}
\newcommand{\medsite}{{\normalfont MedWeb}}

\newcommand{\nsc}{NSC}

\newcommand{\sid}{$sid$}
\newcommand{\flow}{$f$}

\newcommand*\circled[1]{\tikz[baseline=(char.base)]{%
            \node[shape=circle,scale=0.8,draw,inner sep=1pt] (char) {#1};}}

\newcommand{\dg}[1]{\textcolor{cyan}{DG: #1}}
\newcommand{\pd}[1]{\textcolor{red}{PD: #1}}

\newcommand{\todo}[1]{{#1}}
\newcommand{\update}[1]{{#1}}
\newcommand{\final}[1]{{#1}}

\theoremstyle{definition}
\newtheorem{thm}{Theorem}
\newtheorem{lem}[thm]{Lemma}

\usepackage{textcomp}

\begin{document}

\title{\Large \bf \sys: Comprehensive Network Side-Channel Mitigation in the Cloud}
\author[1,2]{Aastha Mehta}
\author[1]{Mohamed Alzayat}
\author[1]{Roberta De Viti}
\author[1]{Bj\"orn B. Brandenburg}
\author[1]{Peter Druschel}
\author[1]{Deepak Garg}
\affil[1]{Max Planck Institute for Software Systems (MPI-SWS), Saarland Informatics
Campus}
\affil[2]{University of British Columbia (UBC)}

\date{}

\maketitle

\begin{abstract}
Network side channels ({\nsc}s) leak secrets
through packet timing and packet sizes. They are of particular concern
in public IaaS Clouds, where any tenant may be able to colocate and
indirectly observe a victim's traffic shape. We present {\sys}, the first
system that eliminates {\nsc} leaks in public IaaS Clouds
end-to-end. It builds on the principled technique of shaping guest
traffic outside the guest to make the traffic shape independent of
secrets by design. However, {\sys} also addresses important concerns
that have not been considered in prior work---it prevents internal
side-channel leaks from affecting reshaped traffic, and it respects
network flow control, congestion control and loss recovery
signals. {\sys} is implemented as a paravirtualizing extension to the
host hypervisor, requiring modest changes to the hypervisor and the
guest kernel, and only optional, minimal changes to applications. We
present {\sys}'s key abstraction of a {\em cloaked tunnel}, describe
its design and implementation,
prove the security of important design aspects through a formal model,
and show through an experimental
evaluation that {\sys} imposes moderate overheads on bandwidth, client
latency, and server throughput, while thwarting attacks based on
state-of-the-art CNN classifiers.

\end{abstract}

\section{Introduction}
\label{sec:intro}



Sharing resources is in the very nature of public Clouds. However, many
side-channel leaks arise when mutually distrusting parties share hardware
resources. Shared CPUs, cores, caches, and memory buses have all been exploited
as side channels~\cite{xu2015, drama2016,
flushreload2014, liu15llcpractical, irazoqui2015ssa, yarom2017cachebleed,
survey2016, netspectre}. As a result, side-channel leaks in Cloud environments
are a growing concern for computer security research.

In this paper, we revisit a specific class of side-channel
leaks---those arising from shared network elements---in the specific
setting of public IaaS Clouds.  These channels, called \emph{network
side channels} ({\nsc}s), leak information via traffic shape (packet
timing and packet size) even when packet payloads are encrypted.  We
argue below that such leaks ought to be a serious concern in public
Clouds. We then describe key requirements for a practical,
comprehensive defense that mitigates {\nsc}s in public Clouds. Despite
decades of work on mitigating {\nsc}s, these requirements have not
received much attention. We present {\sys}, a new system
that satisfies all the requirements and effectively mitigates {\nsc}s
in IaaS Clouds.

\smallskip\noindent
\textbf{{\nsc}s are a serious concern in Clouds.~} Prior work has
shown that traffic shape is strongly correlated with secrets in many
applications -- traffic shape can reveal sensitive information about
webpages~\cite{and98trafficanalssl, sun02statistical, hintz2002,
  cai2012touching, gong2012website,
  dyer2012peekaboo, wang2014effective, hayes2016k, 
  li2018measuring}, video streams~\cite{beautyburst}, VoIP
chats~\cite{wright2008spot}, users' keystrokes~\cite{song2001timing},
and even private keys~\cite{brumley2005remote, brumley2011remote}.
Chen \emph{et al.}~\cite{chen10reality} demonstrate that users'
medical conditions, family income, and investments can be
gleaned from the encrypted traffic of healthcare, taxation,
investment, and web search services provided as software-as-a-service
(SaaS) offerings.

While many of these attacks relied on direct access to the victim's
traffic, more recent work has shown that an \emph{unprivileged
adversary} can also indirectly infer the victim's traffic shape by
inducing contention with the victim's traffic at a shared network
element and measuring resulting variations in the adversary's own
traffic shape~\cite{agarwal2016moving, getoffmycloud, beautyburst}. In
fact, we were able to create such an indirect attack to recognize
streamed videos with 96\% accuracy using a CNN classifier
(\S\ref{sec:attack}).  Such indirect attacks are of particular concern
in public (IaaS) Clouds as adversaries can rent virtual machines (VMs)
and even colocate with a victim's VM at low cost~\cite{getoffmycloud,
  inci2015seriously, inci2016efficient}. Hence, {\nsc}s \emph{should
be a significant concern} for security researchers, Cloud tenants and
Cloud providers alike.

\smallskip\noindent \textbf{Requirements for mitigating {\nsc}s.~} Any
comprehensive mitigation of {\nsc} attacks in an IaaS Cloud must
satisfy the following requirements. {\bf R1.} The mitigation
must prevent leaks through all aspects of \update{the shape of transmitted
traffic}, with provable
guarantees. {\bf R2.}  In line with basic Cloud philosophy, the
mitigation must allow for elastic (dynamically adaptive) sharing of
network resources. {\bf R3.}  Although some overhead for mitigating a
threat as strong as {\nsc}s is unavoidable, the mitigation should
still permit responsive, client-facing services and not require excessive
resources. \mbox{{\bf R4.} The} mitigation should work with any guest VM and
accommodate bursty network traffic, with minimal application changes.
In other words, the mitigation should be \emph{general}.

These requirements rule out many {\nsc} mitigation techniques, specifically
those that prevent leaks via either packet timing or packet size but not both
(violates R1)~\cite{wu2015deterland, wright2009morphing, shan2021real}, do not
handle bursty traffic (violates R4)~\cite{wright2009morphing, dyer2012peekaboo},
\update{rely on multipath routing~\cite{de2020trafficsliver} or} adding
``best-effort'' noise without strong guarantees (violates
R1)~\cite{kocher1996timing, HTTPOS}, hard, static bandwidth reservation for
tenants including TDMA (violates R2)~\cite{vattikonda2012practical},
or application code rewriting (violates R4)~\cite{HTTPOS}.

A general approach that can meet these requirements is to change the
\emph{shape the traffic} in a dedicated system component \emph{outside
the application} to make it independent of secrets. The final shape
can be learnt by the shaping component
adaptively~\cite{lu2018dynaflow,askarov2010}, or the application can
provide it to the shaping component~\cite{cai2014csbuflo,
  wang2017walkie}.
Although this approach has been considered in prior work, the Cloud
setting and the public Internet have additional practical requirements
that have \emph{not been considered} in prior work: {\bf R5.} The
traffic-shaping logic must take flow control and loss recovery of the
network protocol into account (else information may leak \update{via
  the presence of ACKs in the reverse direction}), and it must respect
network congestion signals (else it could destabilize the
network). {\bf R6.} The traffic-shaping component must be integrated
with Cloud servers---as opposed to routers or middleboxes---to prevent
colocated attacker from exploiting contention on servers' network
interface (NIC). Consequently, the shaping component must be
performance-isolated from secret-carrying Cloud tenants to prevent
\emph{internal side-channel leaks within the server} from affecting
the reshaped traffic.
To the best of our knowledge, no prior work on {\nsc} mitigation
satisfies all of R1--R6.

\smallskip\noindent
\textbf{Our contribution: {\sys}.~}
The requirements R1--R6 pose significant design and engineering challenges
for a secure and practical {\nsc} mitigation solution.
%
To our knowledge, {\sys}
is the first end-to-end system that mitigates {\nsc}s
comprehensively addressing all requirements.
{\sys}'s contribution is twofold:
a novel {\em cloaked tunnel} abstraction that shapes traffic between two guests on
different hosts end-to-end, and
a realization of this abstraction for IaaS Clouds.

Briefly, a cloaked tunnel
shapes application traffic to provably make it independent of secrets
at the traffic's origin (R1). This eliminates {\nsc}'s by design. The
tunnel multiplexes multiple flows at fine granularity (R2). The tunnel
works with all IaaS VMs and unmodified applications, although, to
improve efficiency, applications may \emph{optionally} interact with
the tunnel to specify which traffic shapes should be used on their
flows to improve efficiency (R3); this requires only small changes to
applications (R4). 
(The tunnel is secure as long as applications pick
shapes independent of secrets.)
Finally, by design, the tunnel is isolated from guest applications (R6) and it
takes network congestion, flow control, and loss recovery into account when
shaping traffic (R5).
The cloaked tunnel described above is a general abstraction that mitigates
{\nsc} leaks in any setting, not just Clouds.

In addition, {\sys} implements a {\em paravirtualized} instance of the cloaked
tunnel integrated with IaaS Cloud servers. 
{\sys} relies on a hypervisor component, called HyPace, and a guest
kernel module, called GPace, which interacts with HyPace to facilitate
congestion management, loss recovery, and flow control during shaping
(R5). HyPace and GPace implement a novel {\em masking} mechanism to
ensure timely packet transmission independent of guest delays, thus
achieving performance isolation from the guests (R6). Furthermore,
HyPace implements a secure {\em batching} mechanism to amortize the
high costs of masking and sustain \textasciitilde7.6 Gbps line rate
(R3).  An experimental evaluation of our prototype on two IaaS
applications---a medical information site and a video streaming
service---shows that {\sys} defeats powerful {\nsc}s with moderate
overhead. 

\smallskip\noindent{\bf Organization.} We present the threat model, design
challenges, and key ideas behind {\sys} in \S\ref{sec:overview}. We describe the
general cloaked tunnel abstraction, define its requirements and properties, and
argue its security in \S\ref{sec:tunnel}. We describe {\sys}'s implementation of
a paravirtualized instance of the tunnel in IaaS Cloud servers in
\S\ref{sec:design}. We discuss generation of efficient transmit schedules for
traffic shaping in \S\ref{sec:profiling}. We present our implementation and
empirical evaluation of {\sys}'s performance overheads and security in
\S\ref{sec:eval}.  We discuss related work in \S\ref{sec:related} and conclude
in \S\ref{sec:conc}.  Additionally, we present NSC attacks under various setups
in \S\ref{sec:attack} and a detailed evaluation of the security of {\sys}'s
masking mechanism in \S\ref{sec:eval-masking}.
Finally, we build an abstract formal model capturing relevant
aspects of Pacer's design and prove its
security in \S\ref{sec:formal-model}.

\section{Overview}
\label{sec:overview}
As a {running example}, we use the scenario of
a patient who consults a trusted, Cloud-hosted medical website
{\medsite} for diagnostic and therapy options, medical procedures, and
care providers in their area. The patient wishes to keep their
condition from employers, health insurers, and other parties for fear
of discrimination. We show how \sys\ can ensure the patient's privacy
by hiding the content they retrieve from colocated tenants and other
network observers, with minimal modifications to {\medsite}, and with
modest overhead in network bandwidth and response~time.

\subsection{Threat model}
\label{subsec:threat-model}

\begin{figure}[t]
    \includegraphics[width=\columnwidth]{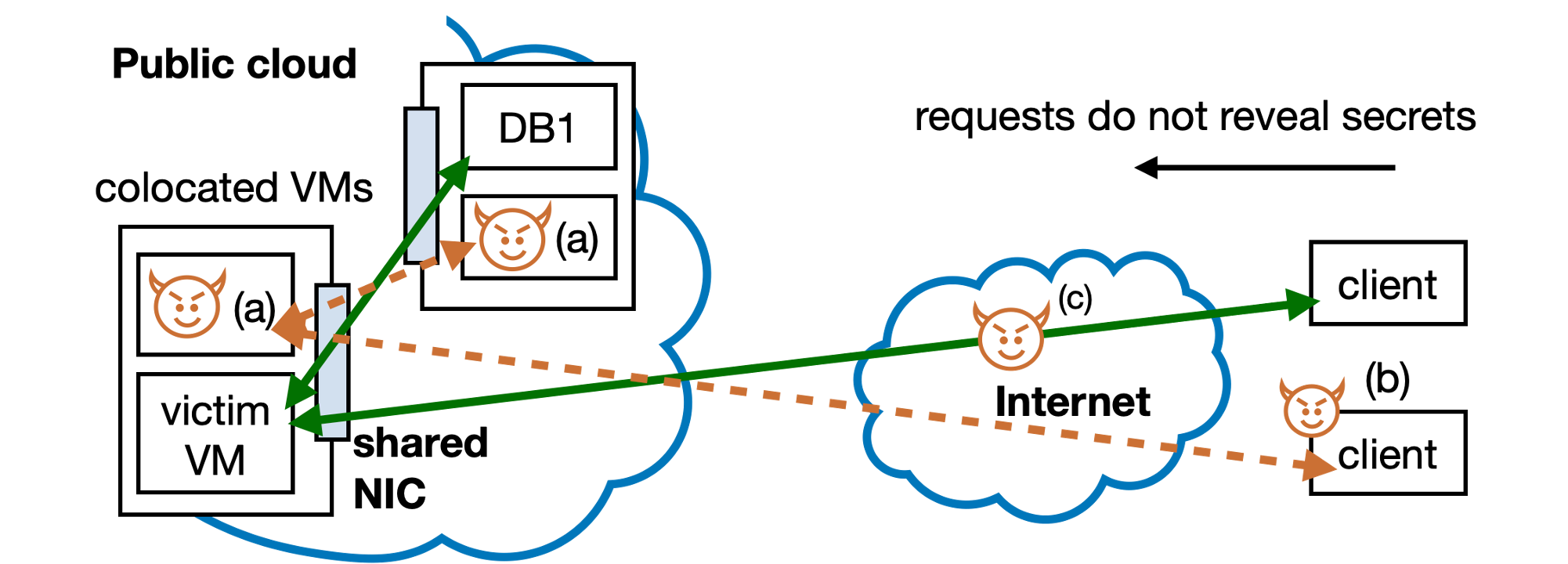}
    \caption{The adversary can (a) colocate with victim's VM or \update{backend
    services} in the Cloud, (b) control clients of its own VMs, and (c) use
    cross-traffic between any pair of these to infer the shape of the victim's
    traffic at shared network links.}
    \label{fig:threat-model}
\end{figure}

The {\em victim} in the public IaaS Cloud is a tenant executing
arbitrary computations in one or more guest VMs, and serving a set of
{\em trusted clients} that connect to its VMs using IPSec or a VPN (virtual
private network)
with pre-shared key authentication\footnote{Guests may require a second level of
  authentication to separate clients' privileges, but this is not
  relevant for \sys's security.}.
The \medsite\ site, for instance, authenticates its registered clients
using IPSec-PSK\footnote{\sys\ requires IPSec-PSK as the
  timing of the tenant's response to an unauthenticated client's
  connection attempt may be affected by tenant's
  concurrent processing of other clients' secrets, thus revealing
  these secrets.}.
\update{To serve a client request, the victim may invoke other Cloud backend
services hosted on separate physical servers.}
The victim's goal is to protect its secrets; these secrets can be
reflected in parameters~of client requests (\eg the name of a
requested file), in the victim's internal state (e.g., which request
handlers are cache-hot because they were recently accessed)\update{,
  or in the backend traffic}.
%

{\sys}'s goal includes preventing {\nsc} leaks of the victim's secrets
to anyone able to rent other VMs in the Cloud. Prior work has shown
that deliberate colocation with a victim VM is
feasible~\cite{getoffmycloud,inci2015seriously,inci2016efficient}. Accordingly,
we assume a strong adversary that may colocate its VMs with the
victim's VM and indirectly infer the shape of the victim's outbound
traffic by observing contention with its own cross-traffic.  The
adversary may use this method to infer the traffic shape of the victim
at shared network elements in the common server, rack or
datacenter\footnote{\update{Prior work has shown that traffic shape
  can be inferred through such methods by an adversary colocated on
  the same physical machine as well as an adversary contending on a
  downstream network link \cite{beautyburst}. Thus, renting dedicated
  physical machines is insufficient to mitigate {\nsc}s.}}.  The
adversary has access to all services available to IaaS guests,
including the ability to time the transmission and reception of its
own network packets with high precision.  The adversary controls
network clients, which communicate with its VMs via the network.
However, the adversary cannot break standard cryptography, break into
the victim's VPN, impersonate/compromise the victim's clients\update{,
  or connect to the instances of backend services used by the victim}.
While not the primary goal of our work, \sys's design also protects
against powerful adversaries who can directly observe the
victim's traffic as well as delay, drop, and inject network packets
(\eg ISPs).
Figure~\ref{fig:threat-model} summarizes {\sys}'s threat model.

\smallskip\noindent
\textbf{Non-goals.~}
\sys\ addresses {\nsc}s; we assume that micro-architectural side-channel
leaks are mitigated by renting an entire server socket and the associated NUMA
domain to the victim for exclusive use.
Alternatively, \sys\ can be combined with complementary work to
mitigate side-channel leaks through other shared
resources~\cite{varadarajan2014scheduler, braun2015robust}.  The Cloud
platform and provider are trusted.

Our focus is on protecting secrets within the content provided by a server.
Hiding the identity of the service requested~\cite{hintz2002, dyer2012peekaboo,
hayes2016k}, the communication protocol used \cite{wright2006inferring,
dyer2013protocol}, or the IP address~\cite{Tor} of the client are non-goals.
{\sys} can be combined with other techniques to address them.
In our running example of \medsite, the patient wishes to hide what specific
disease, procedures and care facilities (s)he is interested in, not that (s)he
is accessing a medical site and video service---most people do occasionally.
Note also that hiding the patient's IP address alone (e.g., using Tor or a
VPN) would be insufficient, because aspects of the content retrieved, e.g.,
the geographic location of care facilities the client retrieves, can reveal
the patient's identity.

\smallskip\noindent
\textbf{Prototype assumptions.~} {\sys}'s
  prototype additionally assumes that {\em clients' request} traffic
  reveals no secrets through its shape (its length, number of packets,
  or timing).
In particular, the time of requests does not depend on any secrets or
the actual completion times of previous responses.
However, {\sys}'s design can support bidirectional traffic shaping,
trivially by running the \sys-enabled hypervisor and kernel, or a kernel with
all of {\sys}'s functionalities on the client side.

\subsection{Key ideas}
{\sys} avoids {\nsc}s by ensuring the shape of the victim's
  network traffic is secret-independent. Ensuring secret-independence
  requires that: (i) The choice of traffic shape must not reveal
  secrets. For instance, if a constant rate of one 1.5kB packet per
  millisecond for 10 seconds is chosen to transmit a particular video,
  then this choice must not be specific to the video.  (ii) If the
  actual packet transmission times deviate from the chosen shape, the
  deviations must not reflect application secrets. In our example, if
  the actual transmission time of a packet deviates from its precise
  expected time based on the rate, then this deviation must not
  reflect the CPU and memory consumption of the concurrent video
  processing in a way that may identify the video.

\smallskip\noindent
\textbf{Secret-independent traffic shapes.~} A strawman secure shaping
strategy is to continuously transmit fixed-size packets at fixed
intervals independent of the application's actual workload; in the
absence of application payload, dummy packets are
transmitted. However, this strategy is highly inefficient when the
application workload is bursty.  {\sys} instead allows the shape to
vary as long as the variations do not depend on secrets.
Specifically, if a guest can partition its
workload into classes based on public information, {\sys} permits the
use of a different, efficient traffic shape in each class.
Returning to our example, suppose {\medsite} streams videos about medical
procedures in different resolutions. Then, the shape used to stream a video can
vary by resolution, which only reveals the patient's available bandwidth but not
the specific medical procedure being watched.
%

\smallskip\noindent
\textbf{Gray-box profiling.~} Secret-independent traffic shaping
requires understanding how a guest's secrets affect its network
traffic.  This information can be obtained by black-box profiling of
guests, but this approach cannot reliably discover all dependencies
and therefore is not secure.  Program analysis, which could discover
all dependencies in principle, does not scale well. \sys\ instead
relies on gray-box profiling, which requires no knowledge of a guest's
internals beyond a {\em traffic indicator} provided by the guest.
This indicator partitions the guest's possible network interactions
independent of secrets, and can be used to profile the guest's network
interactions and generate a transmit schedule for each partition
(\S\ref{sec:profiling}).

\smallskip\noindent
\textbf{Paravirtualized cloaked tunnel support.~} To enforce traffic shapes,
{\sys} provides
paravirtualized hypervisor support that enables guests to implement a
{\em cloaked network tunnel}, while adding only a modest amount of
code to the hypervisor.  A {\em performance-isolated} shaping
component in the hypervisor, called HyPace, initiates transmissions based on a
schedule.
A guest kernel module, GPace, shares state with HyPace for schedule installation
and adaptation based on network congestion, loss recovery, and flow control.
HyPace's and GPace's execution can experience interference from the guest due to
side channels, so {\sys} uses a novel idea---it {\em masks} any execution delays
in HyPace and GPace that could influence actual packet transmission times (R6).

\section{Cloaked tunnel}
\label{sec:tunnel}

In this section, we describe an idealized realization of the {\em cloaked
tunnel} abstraction
and security properties, independent of a specific application setting,
implementation, or placement of tunnel entry and exits.
We begin with a discussion of three high-level properties required in a
practical, secure cloaked tunnel. These properties mirror the requirements R1,
R5, and R6 from \S\ref{sec:intro}.
{\bf P1.}~{\em Secret-independent traffic shape:} Requires that
transmissions follow a schedule that does not depend on secrets (R1),
and that actual transmissions within a schedule are not delayed by
potentially secret-dependent computations (R6).
{\bf P2.}~{\em Unobservable payload traffic:} The traffic shape must not
reveal, directly or indirectly, an application's actual time and rate
of payload generation and consumption. This implies that flow control
must not affect the traffic shape (R5); that padded content must
elicit the same response (e.g., ACKs) from receivers as payload data;
and that packet encryption must encompass the padding. This in turn
requires that padding be added at or above the transport layer, while
encryption be done below the transport layer.
%
{\bf P3.}~{\em Congestion control:} The tunnel must react to network
congestion (R5). Congestion control is needed for network stability
and fairness, but does not reveal secrets since it reacts to network
conditions, which themselves depend only on shaped and third-party
traffic.

\begin{figure}[t]
    \includegraphics[width=\columnwidth]{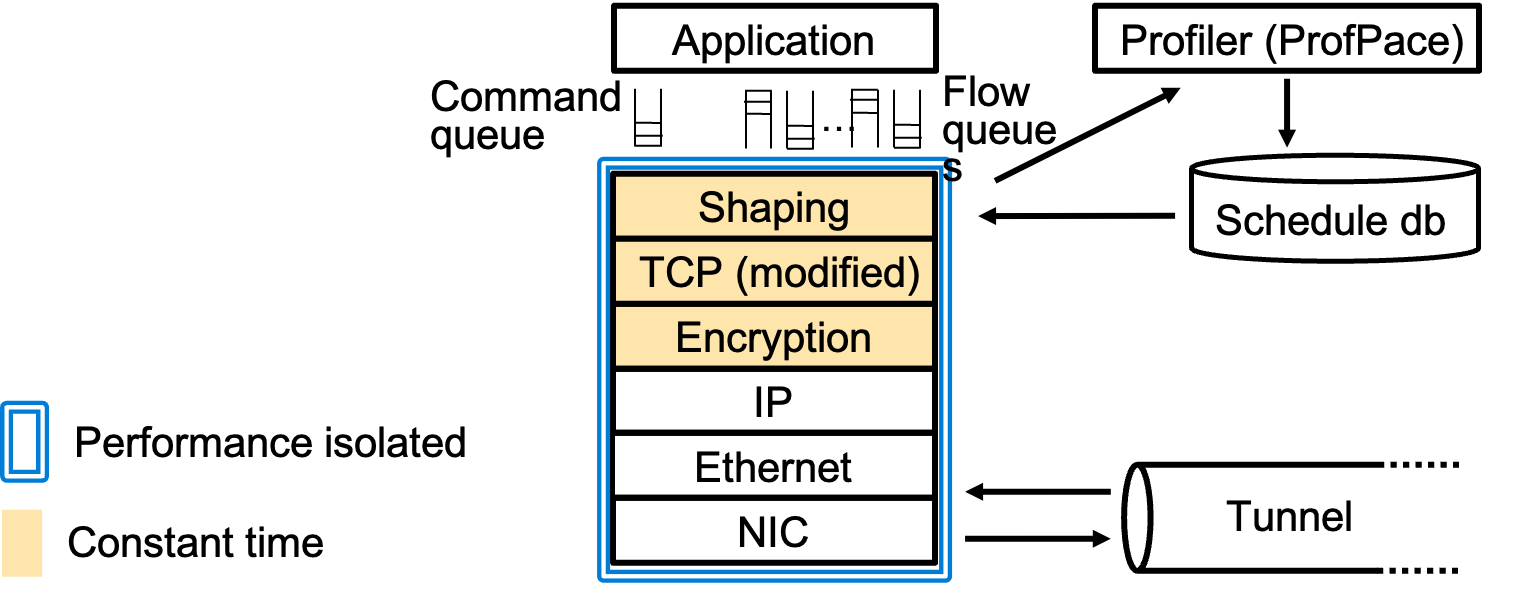}
    \centering
    \caption{Cloaked tunnel (one endpoint)}
    \label{fig:tunnel}
\end{figure}

\subsection{Idealized tunnel design}
\label{ssec:idealized}
Figure~\ref{fig:tunnel} shows the cloaked tunnel's architecture. The
tunnel protocol stack runs on both tunnel endpoints. (Only one of two
symmetric endpoints is shown in the figure.)  The 
stack consists of a {\em shaping} layer on top of a modified transport
layer (e.g., TCP) on top of the encryption layer, e.g. IPSec.
These layers rest on conventional IP and link
layers. Each tunnel is associated with a flow identified by a 5-tuple
of source and destination IP addresses and ports, and the transport
protocol.\footnote{We describe the tunnel in terms of TCP; however,
  another stack like QUIC~\cite{langley2017quic} can also be used.}

The shaping layer initiates transmissions according to a schedule and
pads packets to a uniform size.  It interacts with applications via a
set of shared, lock-free queues. The layer takes application data from
a per-flow outbound queue and transmits it in the tunnel. It places
incoming data from the tunnel into a per-flow inbound queue. Finally,
it receives traffic indicators and per-flow cryptographic keys (to be
used by the encryption layer) via a per-application command
queue.

Overall, the shaping layer shapes each flow separately, and then multiplexes the
shaped traffic of all flows onto the same physical links at the granularity of
packets (R2 of~\S\ref{sec:intro}).

A user-level gray-box profiler, ProfPace, analyzes
timestamps and traffic indicators collected by the tunnel, and generates and
stores
transmit schedules in a {\em schedule database}~(\S\ref{sec:profiling}).

\smallskip\noindent \textbf{Assumptions.~} The tunnel design presented
here relies on some idealized assumptions, which are relaxed in the
practical design of \S\ref{sec:design}. To ensure that packet
transmissions are not delayed by secret-dependent contributions
(property P1), the design
assumes that processing delays in the tunnel network stack are not
influenced by secrets, even indirectly.  This requires that: (i) the
tunnel's layers---especially the shaping, transport, and encryption
layers, which operate on cleartext data---execute in {\em constant
  time}, \ie they avoid data-dependent control flow and memory access
patterns;
and (ii) the execution of the tunnel network stack is {\em
  performance-isolated} from the application and any other
computation.

\smallskip\noindent
\textbf{Outbound data processing.~}
A timestamp is taken whenever data is queued by the application; these
timestamps and the recorded traffic indicators are shared with the
gray-box profiler.  The shaping layer retrieves a chunk of available
data from the flow's outbound queue whenever a transmission is due on
a flow according to the active schedule (if any) and TCP's congestion
window is open (see transport layer below).  The layer removes a number
of bytes that is the minimum of (i)~the available bytes in the queue,
(ii)~the receiver's flow control window (see transport layer below),
and (iii)~$M$, the network's maximal transfer unit (MTU) minus the
size of all headers in the stack.  If fewer than $M$ bytes (possibly
zero) were retrieved from the queue due to payload unavailability or
flow control, the shaping layer pads the chunk to $M$ bytes. It adds a header
to indicate the amount of padding added.

\smallskip\noindent
\textbf{Transport layer.~}
The transport layer operates as normal, except for two
tunnel-related modifications to satisfy R5:~(i) When the congestion window closes, the transport layer signals the
shaping layer to suspend the flow's transmit schedule until the
window reopens. Schedule suspension ensures network
stability and TCP-friendliness, and does not leak information because
it depends only on network conditions, which are visible to the
adversary anyway.
%
(ii) Flow control is modified to make it unobservable to the
adversary. The transport layer signals to the shaping layer the size
of the flow control window advertised by the receiver. This window
controls how much payload data is included in packets generated by the
shaping layer (see outbound data processing above). The transport layer transmits packets
irrespective of the flow control window, sending dummy packets while
the window is closed, which are discarded at the tunnel's other end.

The transport layer passes outbound packets to the encryption layer,
which adds a message authentication code (MAC) keyed with the flow's
key to a header and encrypts the packet with the flow's key. Finally,
encrypted packets are passed to the IP layer, where they are processed as normal
down the remaining stack and transmitted by the NIC.

\smallskip\noindent
\textbf{Inbound packet processing.~}
Packets arriving from the tunnel are timestamped; the stamps are
shared with the profiler. Packets pass through the layers in reverse
order, causing TCP to potentially send ACKs. The encryption layer
decrypts and discards packets with an incorrect MAC. The shaping layer
strips padding and places the remaining payload bytes (if any) into
the inbound queue shared with the application.

\smallskip\noindent
\textbf{Schedule installation.~}
A transmit schedule must be installed on a flow before data can be
sent via the tunnel.
A schedule is associated with each flow's 5-tuple {\flow}
and a traffic id {\sid}, and can be of two types: {\em default} and {\em
custom}. A default {\sid} maps to a default schedule that is installed when the
flow is created. This schedule acts as a template, which is instantiated
automatically by the shaping layer whenever a packet arrives that
indicates the start of a new network exchange (\eg a GET request on a persistent
HTTP connection), identified by the TCP PSH flag. The schedule starts at the
arrival time of the packet that causes the schedule's instantiation.

A default schedule active on a flow can be optionally extended by a custom
schedule in response to an application's {\sid}.  For
instance, a default schedule that allows a TLS handshake might be
extended with one that is appropriate for the response to the first
incoming network request.
The shaping layer looks up the schedule associated with {\sid} in the
schedule database and associates it with flow {\flow}.
Every custom schedule has a sufficiently
large initial delay to allow the schedule to reach the tunnel endpoint
before the first scheduled transmission, despite any queueing
delays. Thus,~the precise time of schedule extension remains
unobservable to the adversary.

\subsection{Tunnel security}
\label{sec:tunnel-security}
The cloaked tunnel provides the following security property: {\em The
  shape of the traffic in the tunnel does not depend on secrets.}
This property holds because our design ensures that each of the
following is either independent of secrets or unobservable to a
network adversary: {\bf S1.} the chosen transmit schedules, {\bf
  S2.} activations, pauses and resumes of transmit schedules, {\bf
  S3.} timing of updates to active transmit schedules, {\bf S4.}
deviations from transmit schedules, and {\bf S5.} transport-layer
responses to transmitted packets.

{\bf S1} is secret-independent by assumption on how schedules are
picked. {\bf S2} and {\bf S5} depend only on public network events
(like congestion signals) by design. {\bf S3} cannot be observed by a
network adversary due to the delay at the beginning of custom
schedules. {\bf S4} is secret-independent because the network stack is
performance-isolated from the application so there are no side-channel
leaks of secrets \emph{into} the network stack, and all processing on
cleartext data in the network stack is constant-time so there are no
internal secret-dependent delays \emph{within} the network stack.

\section{\sys\ design}
\label{sec:design}

We now describe {\sys}, a practical cloaked tunnel design in the
context of a public IaaS Cloud.
%
We first discuss constraints on the design space in the context of an
IaaS Cloud. First, {\em the tunnel entry must be integrated with the
  IaaS server}.  In an IaaS Cloud, colocated tenants typically share
the network link attached to the server and can therefore indirectly
observe each others' traffic.
Therefore, the tunnel entry must be in the IaaS server to ensure the
attached link lies inside the tunnel.
Second, {\em shaping requires padding, which must be done at the transport
layer to ensure it is unobservable}.
Third, the conceptual tunnel design of \S\ref{sec:tunnel} requires
that the {\em network stack is performance-isolated from
  secret-dependent computations and layers that deal with cleartext
  are constant-time}.  All guest computation must be assumed to be
secret-dependent in an IaaS server, suggesting that shaping should be
implemented in the IaaS hypervisor, where it can be executed with
dedicated resources and tightly controlled.


One way to meet these requirements is to place the entire network
stack in the hypervisor, performance-isolate it from guests, and
implement it as constant time. However, this approach has many
limitations. First, ensuring performance isolation for an entire
network stack is technically challenging even in the hypervisor.
Second, implementing the tunnel layers as constant time is not
trivial.  Third, it defeats NIC virtualization, such as
  SR-IOV, and requires guests and their network peers to use the
network stack provided by the IaaS platform. Lastly, it complicates
the hypervisor significantly.


\smallskip\noindent
\textbf{{\sys} architecture.~}
\sys\ addresses the tension outlined above using a paravirtualization
approach.
The responsibilities are divided between the hypervisor and the guest OS such that
{(i) the hypervisor can ensure tunnel
  security with~only weak assumptions about a guest's rate of progress;
  (ii) the performance-isolated hypervisor component is small; (iii)
  the guest OS changes are modest}.
%
We extend the hypervisor to provide a {\em small set
  of functions that allows guests to implement a cloaked
  tunnel}, while guests retain the flexibility to use custom network
stacks on top of a virtualized NIC.

\begin{figure}[t]
  \includegraphics[width=\columnwidth]{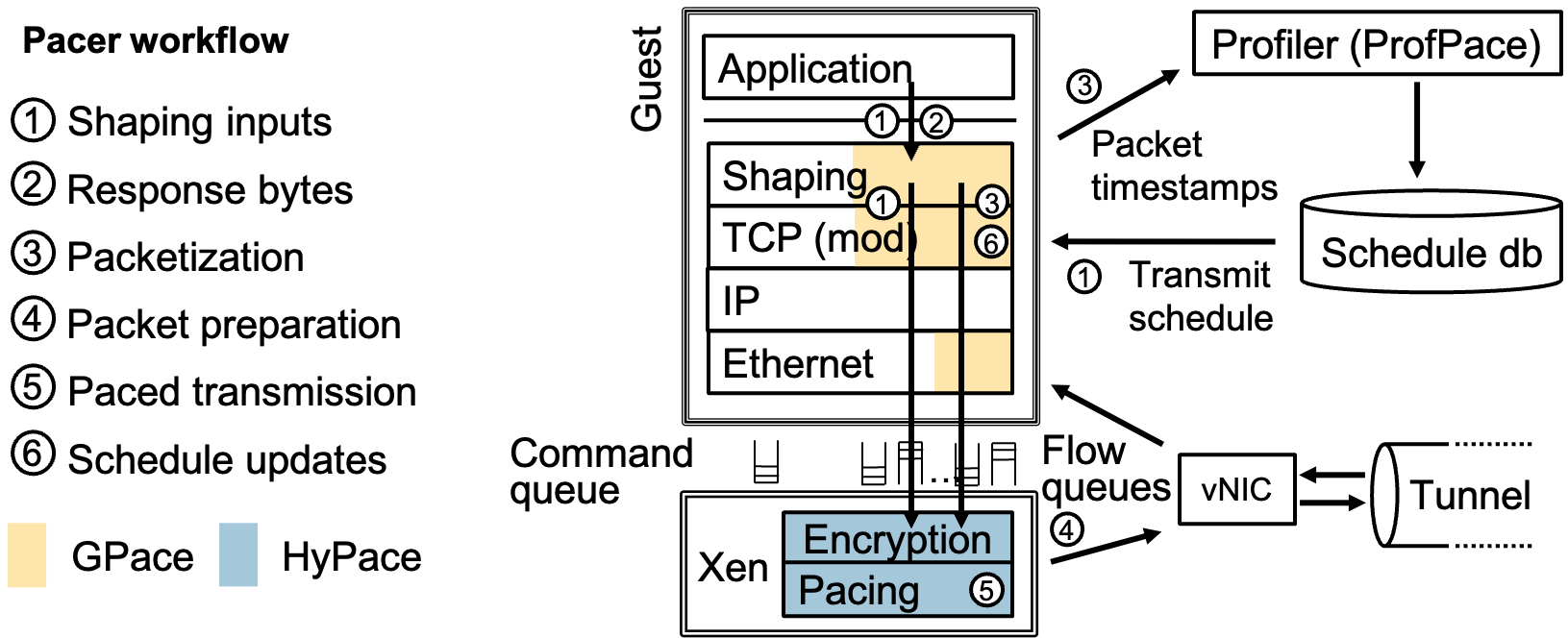}
    \centering
    \caption{\sys\ architecture and workflow.
    }
    \label{fig:pacer-arch}
\end{figure}

Figure~\ref{fig:pacer-arch} shows \sys's architecture and workflow.  Unlike the
strictly layered tunnel stack from \S\ref{sec:tunnel}, \sys\ factors
out a small set of functions that inherently require
performance-isolation into the lowest layer, implemented in the IaaS
hypervisor.  The {\em HyPace} component plugs into Xen and provides
these functions.  The {\em GPace} component, a Linux kernel module,
plugs into the guest OS and the OS of network clients that
interact with the guest. It implements the cloaked tunnel in
cooperation with HyPace.

The guest has direct access to a SR-IOV virtual NIC (vNIC) configured by the
hypervisor, which it uses to receive but not to transmit packets.
When the guest application receives a request, it sends a traffic indicator to
GPace, which shares the indicator along with the flow's 5-tuple, TCP sequence
number, congestion window, and crypto key in a per-flow datastructure to HyPace.
HyPace instantiates a transmit schedule based on the
indicator~\protect\circled{1}.
When the application forwards response bytes to GPace~\protect\circled{2}, GPace
splits the payload into MTU-sized packets with necessary padding, placing them
in the per-flow structure, and timestamps the outgoing packets, sharing the
flow's traffic profile with ProfPace~\protect\circled{3}.
GPace also generates retransmission packets for both payload and dummy packets.
At a scheduled transmit time, HyPace picks a payload packet or generates a dummy
packet, encrypts and adds MACs to the packet, and places it in the NIC
transmission queue~\protect\circled{4}.
HyPace initiates NIC transmission after {\em masking} potentially
secret-dependent delays in its execution~\protect\circled{5}.
Additionally, it adapts the schedules in response to network events (\eg
congestion and retransmission) based on GPace's signals~\protect\circled{6}.

By generating dummy packets subject to congestion control and independently from
the guest network stack, {\sys} requires performance-isolation only for HyPace
but not the guest.
Overall, \sys's security properties remain equivalent to those of the conceptual
cloaked tunnel design (\S\ref{sec:tunnel-security}),
as we discuss in \S\ref{sec:pacer-security}.
%

\subsection{HyPace}
\label{sec:hypace}

Similar to the shaping layer in the conceptual tunnel design, HyPace
receives traffic indicators from applications (via GPace),
instantiates template schedules in response to incoming packets
(signaled by GPace), and initiates transmissions.  To ensure tunnel
security despite potentially secret-dependent delays in the guest,
however, HyPace performs additional functions and there are
differences, which we discuss next.

HyPace implements padding, encryption, congestion control, \update{and
retransmissions} in
cooperation with the guest.  HyPace pauses a transmit schedule when a
flow's congestion window
closes and resumes the schedule when it reopens. When a transmission
is due on a flow and the congestion window is open, HyPace checks
whether the guest has queued a payload packet. If not, it generates a
dummy packet with proper padding, transport header, and encryption,
using the next available TCP sequence number and the flow key shared
with the guest. Next, it initiates the transmission of the payload
or dummy packet, reduces the congestion window accordingly, \update{and
initiates a retransmission timeout for the packet}.
\update{Finally, in case of retransmissions (either due to a timeout on expected
ACKs or due to receiving duplicate ACKs), HyPace extends the transmit
schedule with a slot for every packet retransmitted by GPace. Unlike the
generic tunnel, where the shaping occurs above the transport layer, this
schedule extension is necessary to enable retransmissions. Note that
extending schedules in response to retransmission events is secure
because retransmissions occur only when there are packet losses in the network,
which are publicly observable.} 
%

\smallskip\noindent
\textbf{Interface with guests.~}
HyPace shares a memory region pairwise with each guest. This region
contains a data structure for each active flow.  The flow structure
contains the following information: the connection 5-tuple associated
with the flow; a sequence of transmit schedule objects; the
current TCP sequence number and the right edge of the congestion
window; the flow's encryption key; and, a queue of packets
prepared for transmission by the guest. Each transmit schedule object
contains the {\sid} and a starting timestamp.
HyPace and the guest use lock-free synchronization on data they share.

\smallskip\noindent
\textbf{Packet transmission.~}
HyPace transmits packets according to the active schedule in the
packet's flow. From a security standpoint, packets need not be
transmitted at the exact scheduled times; however, any deviation
between scheduled and actual time must not reveal secrets.

On general-purpose server hardware, it is challenging to initiate
packet transmissions such that their timing cannot be influenced by
concurrent, secret-dependent computations. Using hardware timers,
events can be scheduled with cycle accuracy. However, the activation
time and execution time of a software event handler is influenced by a
myriad of factors. These may include (i) disabled interrupts at the
time of the scheduled event; (ii) the CPU's microarchitectural, cache,
and write buffer state at the time of the event; (iii) concurrent bus
traffic; (iv) frequency and voltage scaling; and (v) non-maskable
interrupts during the handler execution.  Many of these factors are
influenced by the state of concurrent executions on the IaaS server
and may therefore carry a timing signal about secrets in those
executions.

\smallskip\noindent \textbf{Masking event handler execution time.~}
HyPace masks hardware state-dependent delays to make sure they do not
affect the actual time of transmissions.
A general approach is as follows.
First, we determine empirically the distribution of delays
between the scheduled time of a transmission and the time when
HyPace's event handler writes to the NIC's {\em doorbell register},
which initiates the transmission.  We measure this distribution under
diverse concurrent workloads to get a good estimate of its true
maximum \update{and update the estimate whenever a new maximum is
  observed at any time during a system's execution.} We relax this
estimate further to account for the possibility that we may not have
observed the true maximum and call this resulting delay
$\delta_{xmit}$. Second, for a transmission scheduled at time $t_n$,
we schedule a timer event at $t_n - \delta_{xmit}$.  Third, when the
event handler is ready to write to the NIC doorbell register, it spins
in a tight loop reading the CPU's clock cycle register until $t_n$ is
reached and then performs the write.  By spinning until $t_n$, HyPace
masks the event handler's actual execution time, which could be
affected~by~secrets.

Unfortunately, the measured distribution of event handler delays has a
long tail. We observed that the median and maximum delay can differ by
three orders of magnitude (tens of nanoseconds to tens of
microseconds).  This presents a problem: With the simple masking
approach, a single core could at most initiate one transmission every
$\delta_{xmit}$ seconds, making it infeasible to achieve the line rate
of even a 10Gbps link. Instead, we rely on {\em batched
  transmissions}.

\smallskip\noindent
\textbf{Batched transmissions.~}
The solution is based on two insights. (i) Our extensive empirical
observations indicate that the instances in the tail  of the event handler delay
distributions tend to occur very infrequently and never in short
succession\footnote{Without the knowledge of Intel CPU internals, it is
difficult to determine the exact cause of the tail latencies, but their
frequency suggests that they may be caused by system management interrupts.}.
As a result, the maximal delay for transmitting $n$ packets in a single event
handler activation does not increase much with $n$.  Hence, we can amortize~the
overhead of masking handler delays over $n$ packets.
(ii) Actual transmission times can be delayed as long as
the delay does not depend on secrets. Hence, it is safe to batch
transmissions.

We divide time into {\em epochs}, such that all packet transmissions
from an IaaS server scheduled in the same epoch, across all guests and
flows, are transmitted at the end of that epoch. An event handler is
scheduled once per epoch. It prepares all packets scheduled in the
epoch, spins until the batch transmission time, and then initiates the
transmission with a single write to the NIC's doorbell register.

Let us consider factors that could delay the actual packet transmission time
once the spinning core issues the doorbell write. Reads were executed before the
spin, so the state of caches plays no role. The write buffer should be empty
after the spin.  Interference from concurrent NIC DMA transfers reflects shaped
traffic and is therefore secret-independent. Similarly, any delays in the NIC
itself due to concurrent outbound or inbound traffic cannot depend on secrets.
However, the doorbell write itself could be delayed by traffic on the memory 
bus, PCIe bus, or bus controller/switch.


\smallskip\noindent
\textbf{Hardware interference and NIC support.~}
A remaining source of delays are concurrent bus transactions caused by
potentially secret-dependent computations. We tried to detect such delays
empirically and have not been able to find clear evidence of them. Nonetheless, 
such delays cannot be ruled out on general-purpose
hardware. A principled way to rule out such interference would require hardware
support.


For instance, a {\em scheduled packet transmission} function provided
by the NIC would be sufficient.  Software would queue packets for
transmission with a future transmission time $t$. At time
$t-\delta_{bus}$, the NIC DMAs packets into onboard staging buffers in
the NIC. Here, $\delta_{bus}$ would be chosen to be larger than the maximal
possible delay due to bus contention. At time $t$, the NIC would initiate
the transmission automatically.  With such NIC support, HyPace
would prepare packets for transmission as usual, but instead of spinning
until $t_n$ it would immediately queue packets with $t=t_n$.  Incidentally,
NIC support for timed transmissions is also relevant for traffic
management, and a similar ``transmit on time stamp'' feature is
already available on modern smart NICs~\cite{napatechSmartNIC}. We
plan to investigate NIC support in future work.

\smallskip\noindent
\textbf{HyPace summary.~} HyPace is a minimal component implemented in the hypervisor,
which is performance-isolated from the guest and enables guests to
implement a cloaked tunnel.  HyPace's careful design masks any
potentially secret-dependent delays in the \update{(re-)}transmission of packets,
obviating the need for a constant-time implementation of any part of
the tunnel's network stack or a performance-isolated guest network
stack. At the same time, the batched transmission design amortizes the
high cost of masking and helps to sustain packet transmission throughput close
to the NIC's line rate.

\subsection{GPace}
\label{sec:gpace}


GPace is a kernel module that implements a cloaked
tunnel jointly with HyPace\footnote{On the client-side, 
GPace terminates the tunnel in the kernel.}. GPace
pads outgoing TCP segments to MTU size and removes the padding
on the receive path. It modifies Linux's TCP implementation to
share its per-flow congestion window and sequence number with HyPace,
and to notify HyPace of retransmissions so that HyPace can extend the
active schedule.
\update{Furthermore, in case of a retransmission, GPace starts with
retransmitting the first unacknowledged TCP sequence number. If this sequence
number is for a dummy, GPace generates a dummy packet and sends it to HyPace,
which eventually transmits it at a scheduled time.}

Note that TCP's flow control window is not advertised to HyPace,
causing HyPace to send dummies if the receiver's flow control window
is closed, as required.  GPace timestamps outbound data arriving from
applications and inbound packets from the tunnel in the vNIC interrupt
handler. All timestamps and recorded traffic indicators are used
by the profiler (\S\ref{sec:profiling}).

GPace allows applications to install session keys and provide traffic
indicators on flows via IOCTL calls on network sockets.  Recall that
applications specify a flow, a traffic id {\sid} and a type as
arguments when indicating traffic.  GPace passes this information into
the per-flow queue shared with HyPace, which uses the {\sid} as an
index to look up the corresponding transmit schedule in the database.

\smallskip\noindent
\textbf{Packet processing.~}
With GPace, the guest OS generates TCP segments as usual, but pads them to the
MTU size
before passing them to the IP layer\footnote{ACKs are not padded as
{\sys} does not need to hide client traffic shape. However, ACKs
are paced to hide guest's interference with their transmission.}.
Instead of queuing packets in the
vNIC's transmit queue, GPace queues them in per-flow transmit queues
shared with HyPace. The guest OS processes incoming packets as usual
by accepting interrupts and retrieving packets directly from its vNIC.
\smallskip\noindent
\textbf{Schedule (re-)activation delays.~}
\label{ssec:delays}
Unlike the conceptual tunnel design, \sys\ processes inbound network
packets and TCP timeouts in the guest, which is not
performance-isolated. Thus, the
delay between two causally related network events $e_1$ and $e_2$ must be made
independent of actual processing delays in the guest, which may otherwise reveal
secrets.

  
There are three relevant causally related pairs of events: 1) The
arrival of the first packet of a request ($e_1$), which triggers the
instantiation of a default schedule with start time equal to $e_1$'s
timestamp, and the subsequent transmission of a packet ($e_2$)
according to the schedule, 2) An incoming ACK ($e_1$) that either
causes a retransmission or opens the congestion window and triggers
the next packet transmission ($e_2$), and 3) a network event ($e_1$)
that sets a timer and subsequently causes a retransmission when the
timer expires ($e_2$).

In each case, GPace uses masking to hide variability in the
  processing time between $e_1$ and $e_2$. Let $\epsilon$ be HyPace's
  epoch length, $\delta_{delay}$ be the guest OS's empirical
  \emph{maximum} inbound packet- and timer-processing time, and
  $\delta = \epsilon + \delta_{delay}$. Then, for (1): GPace requires
  that the initial response time of any default schedule be larger
  than $\delta$; for (2): GPace ensures that $e_2$ is scheduled no
  earlier than $\delta$ after $e_1$; for (3): GPace ensures that $e_2$
  is scheduled no earlier than $\delta$ after the timeout. These rules
  make the guest's actual processing time between causally related
  network events unobservable to the adversary.

  \smallskip\noindent
  \textbf{GPace summary.~} GPace is a Linux
  kernel module that implements both ends of a cloaked tunnel, using
  the paravirtualized support from HyPace. It handles padding in
  payload packets, shares outgoing packets with HyPace along with
  per-flow sequence numbers and congestion window state, signals
  HyPace on installation of a new transmit schedule or update of a
  transmit schedule, and masks processing delays between pairs of
  causally related network events.

\subsection{{\sys} security}
\label{sec:pacer-security}

We built an abstract formal model of HyPace, the guest and the
network, covering essential details such as delays due to internal
side channels and HyPace's schedule replacement. We formally proved
that our design provides the standard, strong security property of
\emph{noninterference}~\cite{DBLP:series/ais/Smith07}---adversaries
learn nothing about guest secrets (in an information-theoretic sense)
despite observing traffic shape. The formal model and the proof are
presented in
\S\ref{sec:formal-model}.

Here, we provide some intuitive justification of \sys's
security. First, \sys's threat model rules out side-channel leaks to
other co-located tenants through shared CPU state, caches, memory
bandwidth and shared Cloud back-end services. Second, it is impossible
to connect to the victim tenant as a (fake) client and elicit even one
response packet because \sys\ requires packet authentication with
pre-shared keys and GPace silently ignores all unauthenticated
packets.
Third, the adversary cannot learn secrets by measuring the shape of
the victim's traffic because, like the cloaked tunnel of
\S\ref{sec:tunnel}, \sys\ ensures that the shape of outgoing traffic
does not depend on secrets.  This holds because {\bf S1}--{\bf S5}
from \S\ref{sec:tunnel-security} are either unobservable or
independent of secrets for \sys\ as well. The only nontrivial
difference is in the secret-independence of {\bf S4}: while the
cloaked tunnel relies on performance-isolation and a constant-time
implementation of the network stack, \sys\ relies on the empirical
delay-masking mechanisms as above.

\update{
Among the empirical \sys\ parameters, only $\delta_{xmit}$ and
$\delta_{delay}$ are relevant for security; all others like the epoch
length and batch size merely affect performance. If actual delays exceed these
two parameters,  the actual runtime of the transmit handler
or the inbound packet/timer handlers could be exposed, which may  be correlated with victim secrets.
}

\update{
However, to exploit this vector, a colocated adversary would have to
first find a way to cause a delay in the execution of these handlers
beyond what was observed during {\sys}'s systematic training phase for computing
the masking delays. This is
difficult because the adversary is unprivileged relative to handler
executions in both the guest kernel and the hypervisor and, hence, limited in its ability to cause these handlers to preempt. Second,
the adversary would have to extract the secret from the observed
run time. This is difficult because the adversary does not generally
know the nature of the correlation between the secret and the observed
run time. The adversary cannot rely on statistical inference since it can
observe only a single instance of a parameter violation (\sys\ updates the
parameter whenever a new maximum delay is observed). We discuss the
security of \sys's masking in
\S\ref{sec:eval-masking}.
}

\section{Efficient transmission schedules}
\label{sec:profiling}

By default, {\sys} can use the same transmit schedule for all of a
guest's network traffic. This approach does not require any support
from tenant applications and is perfectly secure.  In practice,
however, tenants can significantly reduce bandwidth and latency
overhead by using different schedules for different partitions of
their workload. As long as those partitions are chosen using public
information, no information is leaked. Here, we discuss how
  tenants can safely partition their workload, and use automatically
  generated, efficient schedules for each workload partition.

\subsection{Traffic indicators}
To use custom schedules, a tenant needs
to provide {\em traffic indicators}.  These indicators are used by
\sys\ to instantiate schedules and, along with other logged
information, can be used by {\em ProfPace} to produce transmission
schedules automatically (ProfPace is explained later).

In more detail, traffic indicators are integer-valued events
that a guest generates at appropriate points in its execution. The
indicators serve two purposes:
(i) They indicate the onset of a sequence of transmissions of the class
corresponding to the integer {\sid} argument. This information is used by Pacer
to instantiate an appropriate transmission schedule for the sequence.
(ii) They delimit semantically related packet exchanges within a network flow,
e.g., a client request from the guest's corresponding response.  The integer
{\sid} value of the indicator identifies the equivalence class of the exchange,
e.g., a TCP handshake, a TLS handshake, or
the workload partition to which the request and its response belong, such as the
video resolution in case of {\medsite}.

\smallskip\noindent
\textbf{Instrumenting guests.~} Instrumenting guests to provide traffic
indicators is straightforward. A guest that responds to client
requests on the network, for instance, simply invokes an IOCTL call on
the network socket
before it sends the response. A client, on the other hand,
  calls IOCTL on a new socket to install a schedule before it sends a
  request.  In \S\ref{sec:eval}, we describe how we instrumented
Apache and the PHP applications we use to provide traffic indicators.
\sys\ ensures that the precise timing of the IOCTL call, which
  could reveal secrets, is not reflected in the start time of a
  transmission schedule. If the schedule is instantiated in response
  to a network request, then the schedule is anchored at the request's
  arrival time (see \S\ref{ssec:idealized}). Otherwise, the schedule
  is anchored at a fixed offset from a public event like the top of
  the hour.

\subsection{Choosing workload partitions}
The tenant provides a {\sid} value with each indicator, which
identifies the workload partition and enables \sys\ to use an
efficient schedule.
For {\em performance}, the tenant's choice of {\sid} values should
partition the guest's network traffic into classes of similar
shape. The lower the variance of shapes in each class, the less the
padding required when a specific network response is generated,
minimizing bandwidth overhead. Returning to video streaming in our
running {\medsite} example, there should be a different {\sid} value
for every resolution, and the application should provide this {\sid}
for all videos of this resolution.

For {\em security}, it is sufficient that the choice of {\sid} does
not depend on secrets.
%
First, certain network interaction patterns are well-known and don't
reveal secrets.  For instance, a network server's traffic typically
consists of a TCP handshake, a TLS handshake, and a variable number of
requests and responses on the established connected, followed by a
connection show-down. Using a different {\sid} for each is safe.
Second, the tenant may partition its request-response workload into
equivalence classes, such that the chosen traffic shape reveals the
class but not the specific object requested within a class.  Returning
to the {\medsite} example, all videos with a given resolution may be
considered a class, for which the same {\sid} is used and therefore an
efficient traffic shape (rate) for that class. If all videos are
available in the same set of resolutions, then the resolution reveals
no information about the content requested.

A tenant may choose to further partition its workload into clusters such that
the cluster of a requested object is public, but not the specific object. We
discuss clustering heuristics next.

\subsection{Clustering}
\label{sec:clustering}

Consider a guest that serves a
corpus of objects with a skewed size distribution. Using a single
schedule for the entire corpus requires padding every object to the
largest one in the corpus, incurring a large overhead. Suppose now the
guest can partition the corpus such that each partition contains
objects of similar size, but revealing the partition of a requested
object is not a secret as would be the case when each partition
contains a sufficient large number of objects.  Now, each object can
be padded to the largest object in \emph{its} cluster, which may
reduce overhead significantly without revealing which object within a
partition is being requested.

Determining what clustering is sufficiently private for a specific
content service given its corpus's size and popularity distribution is
beyond the scope of this paper.  We merely highlight here the large
efficiency gains possible when clustering content with skewed size
distributions.
%

We describe heuristic clustering algorithms for videos
and static HTML documents that minimize overhead subject to a \emph{given} privacy
need, which is defined in terms of the minimum number of
objects in any cluster.


\smallskip\noindent
\textbf{Video clustering.~}
We cluster videos according to the sequence and sizes of their 5s
segments using the following algorithm.
Note that dynamically compressed segments differ in size.
Initially, we over-approximate the shape of each video $v_i$ by its
maximal segment size $smax_i$ and its number of segments $l_i$. For
each distinct video length $l$ and each distinct maximal segment size
$s$ in the entire dataset, we compute the set of videos that are
dominated by $\langle l, s \rangle$. A video $v_i$ is dominated by
$\langle l,s \rangle$ if $l_i \leq l$ and $smax_i \leq s$.

Let $c$ be a desired minimum cluster size.
Our algorithm works in rounds. In each round, we select every $\langle l,s 
\rangle$
dominating at least $c$ videos, and we choose as a cluster the set of
videos minimizing the average relative padding overhead per video, \ie
$ \frac{1}{c_i} \sum_{j=1}^{c_i} \sum_{k=1}^{l_i} \left(\frac{(s_k -
      s_{kj})}{s_{kj}}\right) $,
where $c_i$ is the cardinality of the set of videos, $l_i$ is the
maximal length across all videos in the set and $s_k$ is the maximal
size of the $k$th segment across all videos in the set
(\ie $ \max_{1 \leq j\leq c_i}(s_{kj}) $).
The sequence of segment sizes $\langle s_1, s_2, ..., s_{l_i} \rangle$ is the
ceiling of the cluster $c_i$.
Once a cluster is formed, its videos are ignored for later
rounds.
The algorithm stops when all videos are clustered.
If the last cluster has less than $c$ videos, it is merged with the one before
it.

\smallskip\noindent
\textbf{Document clustering.~}
Unlike videos, HTML documents contain a single data object.
Therefore the algorithm clusters based on the single size parameter
of documents, and the largest document in a cluster constitutes the
cluster's ceiling.

More sophisticated clustering algorithms that account for distinct per-object
privacy requirements (popularity) and overall privacy requirements are left to
future work.
We present overheads of clustering real videos and
documents~in \S\ref{sec:padding}.


\subsection{ProfPace}
\label{sec:profpace}

The ProfPace gray-box profiler automatically generates a transmit schedule for
each traffic class as follows.
GPace (\S\ref{sec:gpace}) logs the application-provided traffic
indicators~(\S\ref{sec:profiling}) along
with the arrival times of incoming packets and the times at which the guest OS
queues packets for transmission, and shares the logs with ProfPace.
ProfPace, a userspace process in the guest, analyzes these logs to
compute transmit schedules. Specifically, ProfPace segregates the logs into
network interaction segments and bins them by different values of traffic
indicators, {\sid}s.
The set of observed segments in a bin are considered samples of the
associated equivalence class of network interactions.

ProfPace characterizes the traffic shape for each class with a set of random
variables:
(i) the delay between the first incoming packet and the first response packet
$d_i$,
(ii) the delay between subsequent response packets $d_s$, and
(iii) the number of response packets $p$.
For each equivalence class of network interactions, the profiler samples the
distribution of these random variables from the segments in the
associated bin.

Finally, ProfPace generates a transmit schedule for {\sid} based on the
sampled distributions of the random variables.
Specifically, it generates a schedule with the 99th percentile of the
initial delay $d_i$ and the 90th percentile of the spacing among subsequent packets
$d_s$.
For the number of packets $p$, ProfPace generates the 100th percentile of the number
of packets and further increments this value by a 10\%.
The choice of percentiles is determined empirically; the schedules thus
generated incur minimal overheads on the peak throughput of web services and
moderate overheads on the client response latencies at the cost of a small
increase in bandwidth overheads.

Note that the choice of transmit schedules is relevant only for
performance, not security. An inadequate schedule could increase delays and
waste network bandwidth due to extra padding, but cannot leak secrets.
For good performance, during profiling runs, the guests should sample the space
of workloads with different values of public and private information, as
well as different guest load levels, so that the resulting profiles capture the
full space of network traffic shapes.


  \smallskip\noindent \textbf{Summary.~} {\sys} enables
    tenants to optionally partition their network traffic into public
    classes, where each class is shaped differently using
    custom transmission schedules for efficiency. To use custom
    schedules, a tenant merely has to provide integer-valued traffic
    indicators at appropriate points in its execution. The indicators
    enable Pacer to instantiate efficient transmission schedules,
    and enable ProfPace to automatically generate efficient
    transmission schedules.

\section{Evaluation}
\label{sec:eval}

\begin{table}[t]\centering\footnotesize
\begin{tabular}{|p{0.13\columnwidth}
                |p{0.428\columnwidth}
                |p{0.30\columnwidth}|}
\hline
\textbf{Library} &
\textbf{Traffic indicator} &
\textbf {Traffic class}\\
\hline
Apache &
Before listen() & TCP handshake \\
\hline
Apache &
Before accept() & SSL handshake\\
\hline
Apache &
Before last SSL handshake msg & HTTP requests\\
\hline
MediaWiki &
After parsing requested page title & Page cluster \\
\hline
Video &
After parsing video title/segment & Video cluster, segment \\
\hline
\update{Video} &
\update{Before connect() to memcache} & \update{TCP handshake} \\
\hline
\update{Video} &
\update{Before get request to memcache} & \update{request (1 MTU)} \\
\hline
\update{Memcache} &
\update{Before listen()} & \update{TCP handshake} \\
\hline
\update{Memcache} &
\update{After parsing video title/segment} & \update{Video cluster, segment} \\
\hline
\end{tabular}
\caption{Locations of traffic indicator instrumentation}
\label{tab:indicators}
\end{table}

We implemented HyPace for Xen and GPace's Linux kernel module in
8,100 and \textasciitilde{15K} lines of C, respectively. We
imported 4,458 lines of AESNI assembly code from OpenSSL to encrypt
packets in HyPace. We implemented ProfPace in 1,800 lines of
Python and 1,200 lines of C.

All experiments were performed on Dell PowerEdge R730 servers with
Intel Xeon E5-2667, 3.2 GHz, 16 core CPU (two sockets, 8 cores per
socket), 512 GB RAM, and a Broadcom BCM 57800 10Gbps Ethernet
card, which were connected to a single 10Gbps switch.
The NIC was configured to export SR-IOV vNICs. We disabled
hyperthreading, dynamic voltage and frequency scaling, and power
management in the hosts, which helps to reduce
variance in execution time and ensures consistent, repeatable results
across different runs.

We run \mbox{Xen 4.10.0} hypervisor on each host, which is assigned
one of the CPU sockets and 40GB RAM. Up to two cores are configured to
execute the HyPace transmit event handler in parallel; flows are
partitioned statically among the HyPace cores.  The guest runs an
Ubuntu 16.04 LTS kernel (version 4.9.5, x86-64) in a VM with 8 cores
and 64 GB RAM, and has access to a vNIC. The VCPUs of the guest VM
were pinned \update{1-to-1} to cores on the second socket of the host CPU,
and we used Xen's `Null' scheduler~\cite{xennullsched} for VM
scheduling.  This is in line with our threat model, which assumes that
guests rent dedicated CPU sockets.
Network clients run a modified Ubuntu 16.04 LTS with
  GPace but no hypervisor.

\begin{figure}[t]
	\includegraphics[width=\columnwidth]{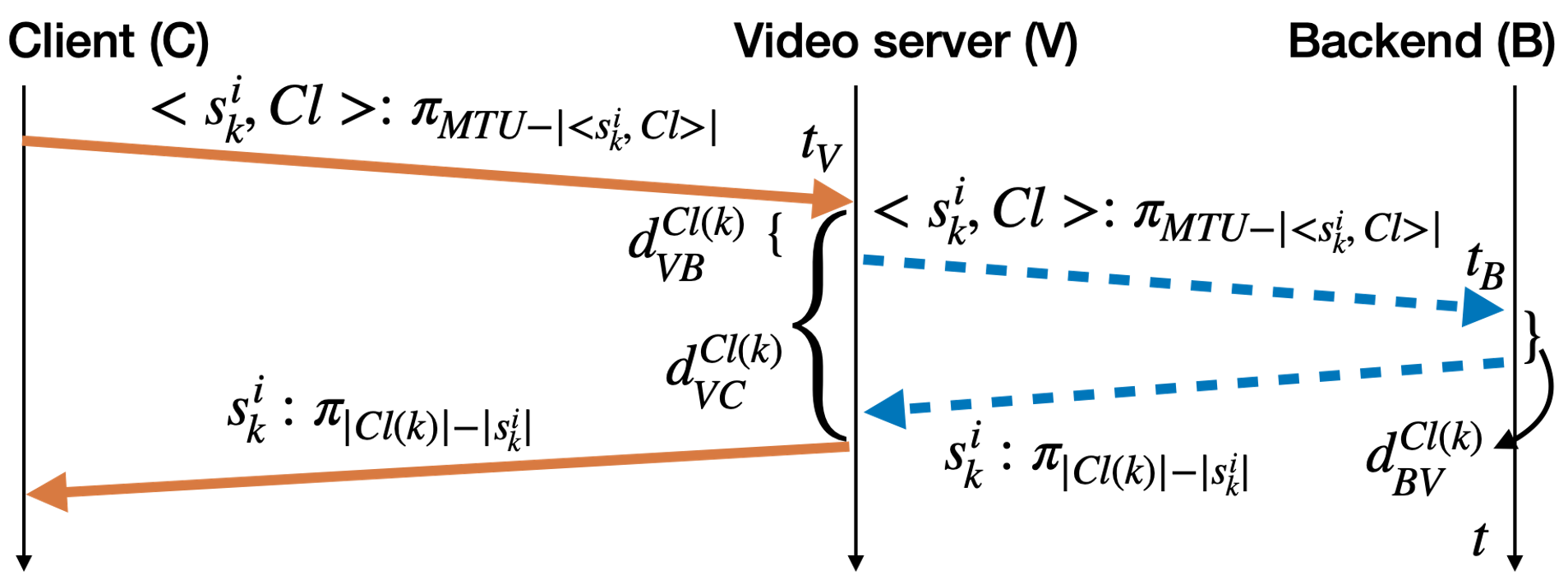}
	\centering
  \caption{\update{Traffic shaping in video service.
  $s^i_k$: $k^{th}$ segment of $i^{th}$ video.
  $Cl$: cluster id of the video with segment $s^i_k$.
  $|X|$: length of payload X.
  ${\pi}_L$: padding bytes of length L.
  $t_V$, $t_B$: arrival time of incoming request.
  $d^{Cl(k)}_{VB}, d^{Cl(k)}_{BV}, d^{Cl(k)}_{VC}$:
  initial delay for $k^{th}$ segment in cluster $Cl$.
  In a {\bf local} setup, shaping is done only on the paths shown by solid
  arrows.
  }}
	\label{fig:videomultitier}
\end{figure}


We used {\sys} to demonstrate {\nsc} mitigation in the
context of a video streaming service and a medical service.
\final{We use {\sys}'s  traffic shaping to hide from an adversary the specific video or
medical webpage requested from the respective service, and we evaluate the
overhead incurred on client latencies and server throughput in the process.
Both the services are hosted using Apache HTTP Server 2.4.33.
Below, we first describe the services and then the modifications introduced in
various applications to generate traffic indicators.}

\smallskip\noindent
\textbf{\update{Video service.~}}
\update{
We wrote a custom video streaming server in PHP, which returns video segments in
response to client requests. We evaluated two setups of the service:
(i) {\bf local:} with the videos hosted on the local VM disk, and
(ii) {\bf 2-tier:} with the videos hosted in a memcached (v1.6.9) KVS backend.
In the {\bf 2-tier} setup, we used one frontend and two replicated
KVSes, each hosted in a VM on a separate server.
The frontend randomly selects either KVS replica for serving each video segment.
For each segment request and response, shaping is done only between the client
and the video server in the {\bf local} setup, while it is done along the entire
network path between the client, the frontend and the KVS selected for the
segment in the {\bf 2-tier setup} (Figure~\ref{fig:videomultitier}).
}


\smallskip\noindent
\textbf{\update{Medical service.~}}
The medical service is \update{a single-tier application} based on Mediawiki (v1.27.1)~\cite{mediawiki1271}.
It stores \update{the content of the} medical pages in a database hosted locally on
MySQL 5.7.16 \update{and caches a copy of HTML pages generated from the content
in a local file cache}. In our experiment, all HTML pages are cached.
\update{When a client requests a page, Mediawiki queries
the database for the page-specific metadata, retrieves the HTML page from the
cache and returns it to the client.
}
%

\smallskip\noindent
Table~\ref{tab:indicators} shows the code locations in the guest applications
where we inserted 15 LoC each to generate traffic indicators. We identified and
modified these sites manually; automating the instrumentation is possible but
remains future work. No other changes were required to guest applications.


%

%


\smallskip\noindent
{\bf Evaluation overview.~}
In the following subsections, we consider
(i) microbenchmarks to determine HyPace and GPace configuration
parameters;
(ii) the tradeoff between spatial padding overhead and privacy possible due to
{\sys}'s clustering,
(iii) {\sys}'s impact on client latencies and server
throughput in the context of two guest applications; and (iv) an
empirical security evaluation of {\sys}'s implementation.

\subsection{Microbenchmarks}
\label{sec:microbenchmarks}
We empirically select the maximum batch size $B$ (number of packets to
be prepared by a HyPace handler) in a suitable HyPace epoch length
$\epsilon$, and the parameters $\delta_{xmit}$ and $\delta_{delay}$
from \S\ref{sec:design}.  To this end, we ran multiple, 12-hour
experiments with varying network workloads.  We requested 100KB-sized
documents from the document server using concurrent clients.  In
background, we ran large matrix multiplications on Xen's dom0 VM, which
used \textasciitilde 12GB RAM and saturated the CPUs.

To determine $\delta_{xmit}$, $\epsilon$ and $B$, we measured the cost
of preparing batches of packets for transmission in HyPace.
Over many observations in the
presence of the background load described above, we first determined
the number of packets that can be safely prepared with different epoch
lengths with a single HyPace handler.
Epochs of length 30$\mu$s, 50$\mu$s, 100$\mu$s and 120$\mu$s could prepare 5,
14, 33 and 42 packets respectively, allowing HyPace to achieve 22\%, 28\%, 41\%
and 42\% of the NIC line rate with one core. We set $\epsilon$ to 120$\mu$s
for all HyPace handlers.

Based on these results, we run two parallel HyPace handlers on two
separate cores. In this configuration, we repeated our measurements
and chose $B$ = 38 packets and $\delta_{xmit}$ = 35$\mu$s for each
handler.
\update{Thus, with two HyPace handlers, {\sys} sustains a line rate of 7.67
Gbps, which is 76.7\% of the NIC's line~rate.}

$\delta_{delay}$ is independent of the number of HyPace threads, and its average
and maximum values observed across all experiment configurations were 3.9ms
and 15.8ms, respectively. We conservatively set $\delta_{delay}$ to 20ms.

\update{Note that only $\delta_{xmit}$ and $\delta_{delay}$ are
security-relevant parameters, which we discuss in detail in
\S\ref{sec:eval-masking}. Epoch and batch size only affect performance.
}

\subsection{Spatial padding overhead}
\label{sec:padding}

\begin{figure}[t]
\centering
\includegraphics[width=\columnwidth]{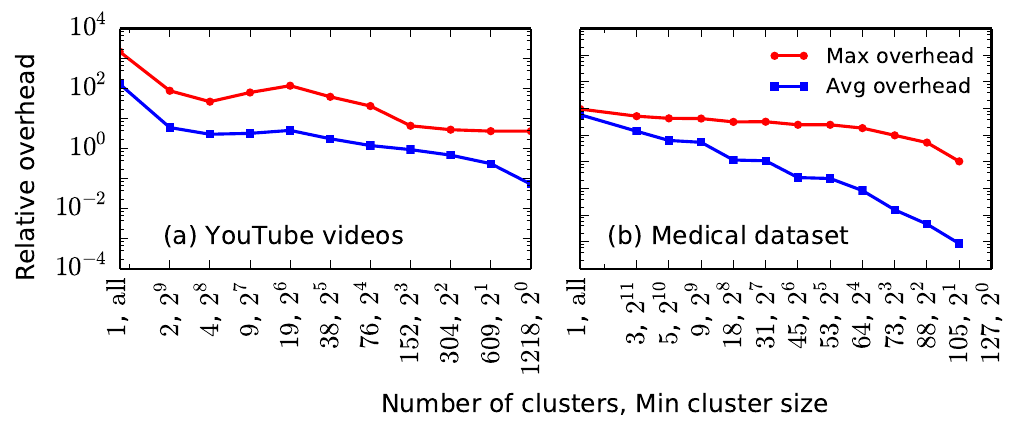}
\caption{Relative padding overhead vs number of clusters and minimum cluster
size for two corpuses representing real-world file size
distributions (log-log scale).}
\label{fig:bwoverhead}
\end{figure}
\begin{figure}[t]
\centering
\includegraphics[width=\columnwidth]{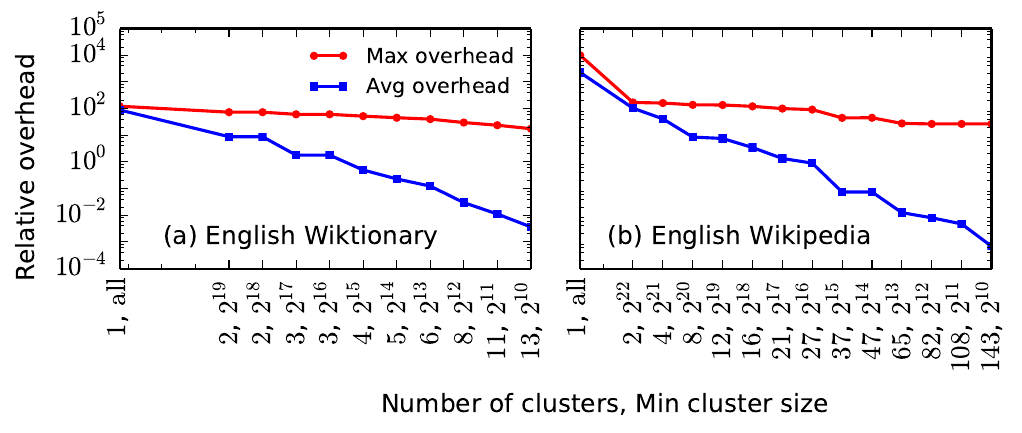}
\caption{Relative padding overhead vs number of clusters and minimum cluster
size (log-log scale) for (a) English Wiktionary and (b) English Wikipedia.}
\label{fig:bwoverhead-extended}
\end{figure}
We measure the tradeoff between spatial padding overhead and privacy
guarantees when clustering content.
The spatial padding overhead corresponds to the network bandwidth overhead for
{\sys}'s traffic shaping.

We clustered two different datasets using algorithms described in
\S\ref{sec:clustering}: (i) a set of 1218 videos downloaded from YouTube
(240p and 720p bitrate, max duration 4.2hr, median duration 7min, max
size 468.7MB, median size 6.2MB), and (ii) a set of 6879 MedicineNet~\cite{medicinenet}
medical web pages comprising
diseases, procedures, medications, and supplements pages (max size
521.9KB, median size 75.2KB).
Figure~\ref{fig:bwoverhead} shows the reduction in the average and
maximum padding overhead with increasing number of clusters and
decreasing minimum cluster size (i.e, the minimum number of objects in
each cluster).  Compared to the medical dataset, the overhead
reduction is less for videos
due to the multi-dimensional clustering~needed for videos.  Nonetheless, even
clustering the corpuses into just two clusters leads to at least two orders of
magnitude reduction in the average padding overhead.

We also compare {\sys}'s clustering with other shaping approaches described in
the literature. Specifically, CS-BuFLO~\cite{cai2014csbuflo} and
Tamaraw~\cite{cai2014tamaraw} round up each
response to the nearest power of 2 and a multiple of some integer value (\eg
$L = 100$ in their paper), respectively.
As can be seen from Table~\ref{tab:clustering}, rounding methods may still leave
files with unique sizes in clusters of size 1, rendering the files immediately
identifiable.
With {\sys}'s clustering, the overheads are comparable even when generating
clusters with more than 2200 files each.
We observe similar results with videos. In fact, the rounding methods of prior work lead
to nearly all the videos in the corpus being in clusters of size 1.
Thus, rounding methods cannot guarantee privacy for all objects in the
corpus, while {\sys}'s clustering can be configured
based on desired privacy requirements and bandwidth constraints.

\begin{table}[t]\centering\footnotesize
\begin{tabular}{|p{0.28\columnwidth}|p{0.07\columnwidth}|p{0.07\columnwidth}|p{0.12\columnwidth}|p{0.12\columnwidth}|}
\hline
{\bf Technique} & {\bf $c_{min}$} & {\bf $n_1$} & {\bf avg OH} & {\bf max OH}
\\
\hline
Power of 2~\cite{cai2014csbuflo} & 1 & 1 & 0.512 & 0.999
\\
\hline
Multiple of 100~\cite{cai2014tamaraw} & 1 & 219 & 0.001 & 0.002
\\
\hline
{\sys} ($c_{min}=1$) & 1 & 37 & 0.009 & 0.027
\\
\hline
{\sys} ($c_{min}=8$) & 8 & 0 & 0.002 & 0.989
\\
\hline
{\sys} ($c_{min}=2206$) & 2206 & 0 & 1.41 & 5.17
\\
\hline
\end{tabular}
%
\caption{Comparison of privacy and overheads in prior work and {\sys}.
$c_{min}$: size of the smallest cluster; $n_1$: number of clusters with a
single element generated by each technique.
{\sys}'s $c_{min} = 1$ is similar to~\cite{cai2014tamaraw} with rounding
up to MTU. 
}
\label{tab:clustering}
\end{table}

\smallskip\noindent
\textbf{Clustering on larger corpuses.}
To understand the impact of padding on larger corpuses, we additionally ran our clustering algorithm on two wiki corpuses:
(i) a 2016 snapshot of the English Wiktionary corpus (5,027,344 documents, max
521.9KB, median 4.7KB), and (ii) a 2008 snapshot of the English
Wikipedia corpus (14,257,494 documents, max 14.3MB, median 83.5KB).
Note that though Wiktionary pages and Wikipedia pages
are not sensitive and may not need protection with a
system like {\sys} in practice, all that matters for our evaluation
are the file sizes and size distributions. The content is irrelevant as it is encrypted
during transmission anyway.
We present the clustering results in Figure \ref{fig:bwoverhead-extended}.

\subsection{Macro experiments}
\label{subsec:eval:macro}
Next, we measure the impact of {\sys}'s traffic shaping on the client response
latencies and server throughput in the video service and medical document service.
The client request payload is only a few bytes and,
hence, is padded to one MTU by the GPace in the client's kernel.
\update{Furthermore, the client is open loop, \ie it transmits
requests to the server at fixed intervals independent of the server's prior
responses.}
\update{ With Pacer support on clients (see
  \S\ref{subsec:threat-model}), the shaping of client traffic will
  similarly ensure that request timing does not depend on the
  completion time of prior requests.
  }

\smallskip\noindent
\textbf{Video service.~}
We wrote a Python streaming client that simulates
a MPEG-DASH player: when a user requests a video, the client initially
fetches six segments (covering 5s of video each) in succession
to fill a local buffer. After reaching 50\% of the initial buffer
(rebuffering goal), the player starts consuming the segments from the
buffer. The client fetches subsequent segments whenever
space is available in the buffer.
We measure the impact of traffic shaping on (i) the download
latency for individual video segments, (ii) the initial delay
until the video starts playing, and (iii) the frequency and duration of any
pauses (video skipping) experienced by~the player.
We use a corpus of videos 1218 videos downloaded from YouTube in March 2018,
which were clustered into 19 clusters with at least 64 elements each, yielding
an average padding overhead of 4x.
The client sequentially plays four randomly chosen videos for up to 5 min each.

We ran experiments for a client with high bandwidth (10Gbps)
\update{and with low bandwidth (10Mbps)}.  \update{The baseline
  segment download latency is <1ms on average, while the exact latency
  depends on the segment size. With {\sys}, the download latency is
  dominated by the initial response latency in each segment's traffic
  shape at the video server, which is 30ms and 400ms in {\bf local}
  and {\bf 2-tier} setup, respectively. Despite these overheads,}
there is no noticeable impact on the user experience for using {\sys}
for either client \update{in either setup}. Initial startup delays,
\update{i.e., the delay until a video starts playing,} don't increase
significantly, and there is no video skipping in any of the
experiments. When serving \update{128} high bandwidth clients in
\update{{\bf 2-tier} setup}, the maximum CPU utilization on the
\update{video} server \update{and the KVS} increases
\update{respectively} from \update{11.76}\% to \update{13.39}\% and
\update{1.96\% to 12.62\%} with {\sys}.


\begin{figure}[t]\centering
    \includegraphics[width=0.6\columnwidth]{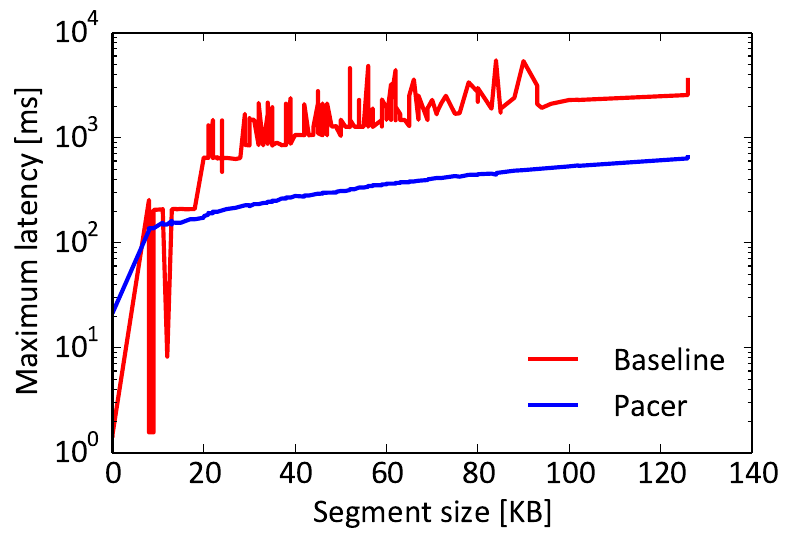}
    \caption{
        Download latency for a 10Mbps client 
    }
    \label{fig:pacer:videortt}
\end{figure}

\smallskip\noindent{\bf Impact on 10Mbps clients.~}
We also evaluated the effect of {\sys}'s shaping on bandwidth-constrained clients
streaming videos.
Here, {\sys}'s shaping also provides an opportunity to use domain knowledge to
optimize schedules for better {\em performance}.
Downloading the largest segment in our collection of 240p videos within its 5s
deadline requires packets to be sent at an interval of max 20ms.
Conservatively increasing the inter-packet spacing in the schedules~to even~6ms
allows downloading segments within 5s. However, for the 10Mbps clients,
the paced schedule avoids losses and reduces the segment
download latency significantly.
Figure \ref{fig:pacer:videortt} shows the download latency for a 10Mbps client for
different segment sizes in the baseline, and after applying {\sys}'s shaping
with 6ms inter-packet spacing. These results are based on the {\bf local} setup
(\S\ref{sec:eval}).
This schedule optimization does not affect security; it only utilizes
{\sys} to reduce network contention, a known benefit of traffic shaping.

\begin{figure}[t]
    \includegraphics[width=\columnwidth]{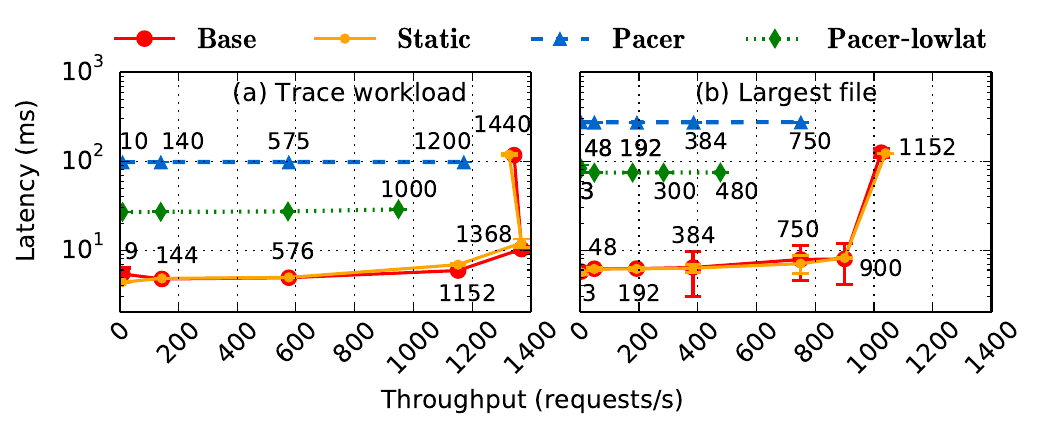}
    \centering
    \caption{Medical service throughput vs latency}
    \label{fig:mwxput}
\end{figure}

\smallskip\noindent
\textbf{Medical service.~}
Next, we measured \sys's impact on the throughput and response latency
of the medical server. We use a corpus of static HTML pages downloaded from
MedicineNet, a medical website \cite{medicinenet}, in August 2020. We used 3 clusters
with at least 2048 elements each, which yields an average padding
overhead of 142.8\%.  Modified wrk2~\cite{wrk2} based clients issue HTTPS GET
requests for different pages concurrently and
synchronously for 120s.  Prior to the measurement, we ran the workload
for 10s to warm up caches.

We selected 1,000 files from the 3 clusters in proportion 65\%, 30\%, and 5\%,
and used this as a workload trace. Each client requests files from the trace in
a random order.
For comparison, we also stressed the server with requests only to the largest
file in the corpus (521.9KB).
Figure~\ref{fig:mwxput} shows the throughput vs average latency for
\update{two insecure} baselines and \update{two schedule configurations of} {\sys} with varying number of concurrent client
requests \update{(denoted by the data point~labels)}. The error
bars show the standard deviations of the average latencies. {\bf Base}
corresponds to the baseline \update{without any shaping}, \update{{\bf Static} corresponds to a baseline where the HTML pages are
statically padded but the response traffic is not paced}, {\bf
    {\sys}} and {\bf {\sys}-lowlat} correspond to {\sys} with schedules
using the 99th and 80th \%ile initial response latency,
respectively.

\update{The performance of {\bf Static} is nearly the same as {\bf Base} because the
padding added is low for the trace workload and zero for the largest file,
implying that there is not much difference in the two workloads.}
Unlike the baseline, \sys's latency remains constant until the maximal
throughput, because latency is determined only by the transmission schedule.
Once the server is at capacity, it fails to serve additional requests and
clients time out.
With {\bf {\sys}-lowlat} a few requests (\eg less than 50 out of 85K) timeout even at
lower loads. \update{This is because {\bf {\sys}-lowlat} uses an aggressive schedule
that does not account for the server's response latencies beyond the 80th \%ile.} We ignore
these timeouts in the average latency and throughput
measurements.

In the trace workload, {\bf {\sys}} incurs a 14.4\% overhead on peak
throughput with response latencies 10x-18x of the baseline. The figures reflect
{\sys}'s total overhead, because they compare to a saturated baseline server. By
comparison, {\bf {\sys}-lowlat} incurs a
\todo{\textasciitilde30\%} overhead on peak throughput with latencies
\todo{3x-5x} of the
baseline. Here, the throughput drops as the requests that are delayed beyond the
80th percentile latency timeout. Similar trends are also observed with the large
file. This shows that latencies could be optimized with moderate additional
overheads on response throughput.

The overheads on peak throughput are higher with the large file.
Here, the baseline operates at over 40\% of the line rate, and we believe that
{\sys}'s performance in this challenging experiment is limited by the accuracy
of transmit schedules, which can be improved substantially.

\smallskip\noindent
\textbf{{\sys}'s costs}
%
are in bandwidth, CPUs, and memory.
The bandwidth overhead (\S\ref{sec:padding}) depends on the application's data,
workloads and the public inputs chosen for workload partitioning.
The bandwidth overhead due to {\sys}'s clustering is comparable to that of prior
work while offering stronger privacy.
The CPU cost is in the two cores dedicated to HyPace
(\S\ref{sec:microbenchmarks}), and the
increase in the guest CPU utilization due to shaping (\S\ref{subsec:eval:macro}).
\sys\ requires less than 20MB of additional main memory in the Xen hypervisor
and less than 30MB of additional memory per HyPace core in each guest that uses \sys.
Cloud providers would likely charge their tenants for the
  added cost of {\nsc} mitigation.  In the case of a service like
  {\medsite}, the tenants (health insurer or provider) would likely
  cover the cost from their own customers' premiums or subscriptions.

\subsection{Security evaluation}
\label{sec:security-eval}


\update{{\sys} is secure by design, as supported by a formal model and
proof in \S\ref{sec:formal-model}.}
Nevertheless, as a sanity check and to validate our prototype implementation,
we also empirically
evaluated the security of {\sys}'s implementation using a powerful
{\nsc} attack. We streamed 4 videos from a single video cluster 40
times each, and collected the precise timestamps and sizes of packets
transmitted in both directions using tcpdump at the video client.
Thus, we grant the attacker direct access to the victim's traffic
shape, which makes the attack more powerful than one launched by a
colocated tenant. (However, as described in \S\ref{sec:attack}, the
attack is effective even when launched by a colocated tenant when
\sys\ is not used.)

\update{For {\sys},}
we trained a multi-feature CNN classifier using timeseries labeled with
the video id and comprising of
inter-packet intervals and sizes of packets in both directions between
client and the server as the features.
\update{For the baseline, we trained the same classifier with just a single
feature---the timeseries of inter-packet intervals in server's response
packets.}
The classifier architecture is similar to that used by Schuster
\etal\cite[\S7.2]{beautyburst},
except
we used a dropout of 0.1 between the model's
hidden layers and 64 epochs for training.

During classification, the classifier generates the probability of each label
value for a given sample.
\update{
In the baseline, the classification probability is more than 99\% for each
label. In {\sys}, it is close to 25\%, \ie the classifier's prediction of a
video's label is no better than a random guess.}
%
We repeated the experiment with other video clusters and obtained
similar results.  Thus, we confirm empirically that, as expected,
{\sys} eliminates leaks through timing, sizes, and count of packets.
\section{Related work}
\label{sec:related}



We compare to existing mitigation techniques and discuss related work
with different threat models or goals.

\smallskip\noindent
\textbf{(a) Mitigating {\nsc}s in Clouds.}
Contention on individual shared links in a Cloud can be mitigated by
time-division multiple access (TDMA) in a
hypervisor~\cite{kadloor2016mitigating,beams2021packet} as this eliminates the
adversary VM's (and, in fact, every VM's) ability to observe a colocated
victim's traffic at that link.  However, an end-to-end mitigation against all
network adversaries would require synchronous TDMA scheduling along the entire
path of a tenant's traffic, which is inefficient especially when the payload
traffic is bursty~\cite{vattikonda2012practical}.
Statistical multiplexing, which only caps the total amount of data
transmitted by a VM in an epoch, is insecure because the
resources available to a flow depend on the bandwidth utilization of
other flows~\cite{gong2016schedulers}.

Another approach restricts the adversary VM's ability to~observe
time~\cite{vattikonda2011eliminating, martin2012timewarp,
  liu2017demand}.
StopWatch~\cite{li2014stopwatch} replaces a VM's
clock with virtual time based only on that VM's execution. To mitigate
{\nsc}s, each VM is replicated $3\times$, the replicas
are colocated with different guests, and each interrupt is delivered
at a virtual time that is the median of the 3 times. This prevents a
guest from consistently observing I/O interference with any colocated
tenant. However, it requires a $3\times$ increase in deployed
Cloud resources. Deterland~\cite{wu2015deterland} also replaces VMs' real time with
virtual time, but it does not address leaks due to {\nsc}s
as it delivers I/O events to VMs without delay.
%
In contrast, {\sys} {\em shapes} traffic by padding and pacing packets, which
  mitigates all {\nsc}s with far less resource overhead.

Bilal {\em et al.}~\cite{bilal2018mitigating} generate multicast traffic to
shape the {\em pattern} of queries to different backend nodes in multi-tier
stream-processing applications in a Cloud, but they do not consider leaks due to
packet size and timing.

\smallskip\noindent
\textbf{(b) Traffic shaping to mitigate {\nsc}s.~}
Pacer uses a standard technique~\cite{hintz2002,wright2008spot} to remove the
dependence of packet {\em size} on secrets: it pads all packets to a fixed
length.
To make packet {\em timing} independent of secrets, a strawman is to send
packets continuously at a \emph{fixed} rate independent of the actual workload,
inserting dummy packets when no actual packets
exist~\cite{saponas2007devices, song2001timing}. This~either
wastes bandwidth or incurs high latencies when the workload is bursty.
%
%
%
BuFLO~\cite{dyer2012peekaboo} reduces this overhead by shaping
response traffic to evenly-spaced \emph{bursts} of a fixed number of
packets~for a certain minimum amount of time after a request
starts. However, it leaks the size of responses that take longer~than the
minimum time.
Tamaraw~\cite{cai2014tamaraw}, CS-BuFLO~\cite{cai2014csbuflo}, and
DynaFlow~\cite{lu2018dynaflow} pad each response to some factor of the
original size, such as the nearest power of 2. They offer no control
over how many objects end up with the same traffic shape. In
  contrast, {\sys} supports flexible traffic shape adaptation without
  leaking secrets. Moreover, as shown in \S\ref{sec:padding}, the
  bandwidth overhead of \sys's clustering is comparable to CS-BuFLO's
  and Tamaraw's.

Walkie-Talkie~\cite{wang2017walkie}, Supersequence~\cite{wang2014effective},
and Glove~\cite{nithyanand2014glove} cluster
responses, and generate a traffic shape for the cluster that envelopes each
response in the cluster.
They cluster by simultaneously considering both packet sizes and timing
from runtime network traces, and compute the shape based on the traces used in
the clusters.
{\sys} instead first clusters based on static object sizes, and then computes
traffic shapes for each cluster based on network
traces of cluster objects. {\sys} can also support clustering and
shaping algorithms proposed by these systems.
Traffic morphing~\cite{wright2009morphing} makes sensitive responses look like
non-sensitive responses, but only shapes packet sizes and ignores packet timing.
{\sys} shapes all packet size and timing,
and allows precise control over cluster sizes,
thus eliminating all leaks by design.

\smallskip\noindent \textbf{(c) Predictive mitigation.~} Predictive
mitigation~\cite{askarov2010, zhang2011predinteractive} mitigates
network timing side-channel and covert channel leaks to an adversary
who has compromised or authenticated as a legitimate client of the
victim. Here, the adversary can distinguish real packets from dummies,
so predictive mitigation cannot avoid a leak when the application
fails to produce a packet in time for a scheduled transmission.
In {\sys}, the threat is from an adversary that only observes network
traffic but does not communicate with the victim. Hence, {\sys} can
hide application delays by sending dummy packets.
Both predictive mitigation and {\sys} partition application workloads based on
public inputs and precompute a traffic shape for each partition.
However, a bad shape leaks information in predictive
mitigation, but only affects performance in {\sys}.

\smallskip\noindent
\textbf{(d) Related work with other security goals.}
Herd~\cite{le2015herd}, ~Vuvuzela~\cite{van2015vuvuzela},
Karaoke~\cite{lazar2018karaoke}, and Yodel~\cite{lazar2019yodel} provide
metadata privacy: they prevent information about who is communicating with whom
from leaking via {\nsc}s.
{\sys}'s goal is different: it prevents sensitive data from leaking via {\nsc}.
To address its~goal, in addition to shaping individual packet sizes and timing,
{\sys} shapes the lengths of application messages.
\update{Herd~\cite{le2015herd} and Yodel~\cite{lazar2019yodel} focus on VoIP
calls. {\sys} can also be used to shape VoIP traffic. For instance, uniform
pacing can be used for a maximum duration, which is picked before the call
from a set of allowed durations. Only this maximum duration, but not the actual
duration, will be leaked.
}

Format-Transforming Encryption (FTE)~\cite{dyer2013protocol}
and ScrambleSuit~\cite{winter2013scramblesuit}
use a tunnel abstraction
to modify payload traffic to bypass a traffic censor's
filters. However, unlike {\sys}, they do not decorrelate the observable traffic
shape from secrets.
%
%
SkypeMorph~\cite{mohajeri2012skypemorph} circumvents censors that
inspect packet sizes and timing.  It samples the inter-packet gap and
the packet size from a fixed distribution, which mimics the
distribution of some target protocol that the censor allows.
SkypeMorph shapes traffic, but unlike {\sys} it is not
  designed to ensure that the resulting shape does not reveal
  secret-dependent variations. Moreover, SkypeMorph transmits traffic
  \emph{continuously} at the average transmission rate of the target
  protocol, which is inefficient for bursty traffic.
Oblivious computing systems~\cite{crooks2018obladi,
  eskandarian2017oblidb, lorch2013shroud} prevent accessed
memory \emph{addresses} or accessed database \emph{keys} from
depending on secrets, for which they rely on ORAM techniques.
{\sys} addresses the orthogonal problem of
making packet size and timing independent of secrets, and relies on traffic
shaping.
Fletcher \emph{et al.}~\cite{fletchery2014suppressing} address \emph{timing} leaks in
ORAM accesses~by pacing ORAM accesses. However, their pacing rate
changes periodically based on the past actual request rate of the
program, which may be secret-dependent and leak information.

\smallskip\noindent
\textbf{(e) Other work.~}
Some prior work~\cite{pu2013your,ongaro2008scheduling,chiang15swiper} use
performance-isolation techniques for performance predictability;
Silo~\cite{jang2015silo} implements traffic pacing to improve remote access
latency; and MITTS~\cite{zhou2016mitts} ``shapes'' memory traffic on CPU cores
for performance and fairness.  The goals and approaches  are
different from {\sys}'s. Richter {\em et al.}~\cite{richter2015hardware} propose
to performance-isolate colocated tenants by modifying the NIC firmware. {\sys}'s
traffic shaping can be implemented in NIC to provide strong isolation
from the rest of the system in the face of microarchitectural side channels
(\S\ref{sec:hypace}).

\section{Conclusions}
\label{sec:conc}

{\sys} is a comprehensive, provably-secure mitigation for {\nsc} leaks in IaaS
Clouds. It reshapes network traffic outside guest VMs to make packet
timing and packet sizes independent of guest secrets. {\sys}
integrates with the host hypervisor to thwart attacks from colocated
tenants, relies on paravirtualization to respect network flow
control, congestion control, and loss recovery, and uses performance
isolation and masking to nullify the effects of internal timing
channels within the host.
{\sys}'s end-to-end overheads are
moderate.



\section*{Acknowledgments}

We thank Lorenzo Alvisi, Bobby Bhattacharjee, Keon Jang, Antoine
Kaufmann, Jonathan Mace, and the anonymous reviewers for their helpful
feedback on earlier versions of this paper. This work was supported in
by part by the European Research Council (ERC Synergy imPACT 610150)
and the German Science Foundation (DFG CRC 1223).



\appendix

{
\bibliographystyle{plain}
\bibliography{pacer}}

\section{Network Side-Channel Attack}
\label{sec:attack}

%
Here, we briefly describe a proof-of-concept {\nsc} 
attack.  To carry out such an attack, an adversary must be able to
observe a victim's network traffic. An adversary with access to
network elements like links, switches, or routers can observe the
traffic {\em directly}.  An adversary without direct access can still
observe victim traffic {\em indirectly} if they can control attack
traffic that shares bandwidth with the victim's traffic.

Indirect observation is impossible if {\em each network flow has
  exclusively reserved bandwidth}, as in time-division multiple access
(TDMA), which ensures non-interference among flows.  However, this
approach prevents statistical multiplexing and is very inefficient for
bursty traffic.  On the other hand, when bandwidth is shared, then
regardless of the queuing discipline, available bandwidth and queuing
delays observed by one flow are influenced by concurrent flows.
We demonstrate a simple attack where an adversary exploits the signals in the
queueing delays for its own traffic to infer the victim's traffic shape.

\smallskip\noindent
\textbf{Experimental setup.}
We set up two VMs, a victim and an attack VM, on two separate sockets
of a Dell PowerEdge R730 server machine ($S_1$). The VMs use Xen's
virtualized network stack; thus all traffic is routed through the
netback driver and the TCP stack in dom0 of the hypervisor. We configure $S_1$'s shared
NIC with a bandwidth of 1Gbps, and the hierarchical token bucket (HTB)
queueing discipline.
We further create two separate HTB traffic classes for (i) the attack traffic, and
(ii) the victim traffic and rest of the traffic through the host. We configure the
attack traffic to have a lower priority than all other traffic. This is a
reasonable assumption as an attacker can always lower the priority for its
traffic.

The victim hosts a custom video streaming service on
top of Apache, which servers video segment files in response to client requests.
A custom video client runs on a second server ($S_2$) and requests the video
segments sequentially over HTTPS.
The attack VM runs a UDP client that sends short
payloads (56 bytes) to a UDP server on a third machine ($S_3$), which
logs the packet arrival timestamps and echoes the packets back to the attack
client.  $S_2$ and $S_3$ have 10Gbps NICs and all machines are
connected with a 10Gbps switch; thus the bottleneck link is the shared
NIC at $S_1$.
The attack client maintains a send window of 4500 packets (computed based on the
bandwidth-delay product for the NIC), which ensures that some attack packets
are always queued at the bottleneck link without overflowing the queue.

We streamed 10 videos at 720p resolution from a YouTube dataset (a
detailed description of the dataset is given in~\S\ref{sec:padding})
for up to 30 segments.
Segments take less than 0.02s to download, and segments within a video are
requested at an interval of 5s.
We streamed each video 150 times.  During each video stream, we log the
series of arrival timestamps of the adversarial client's packets at the
adversarial server. We label each time series of the adversary's packet arrival
timestamps with the id of the video streamed by the victim.  Thus, we have
1500 time series of adversary's packet arrival timestamps with 10
distinct labels.

\textbf{Analysis.}
We aggregated each time series into windows of 50ms, and generated a
time series of the adversary's transmitted packet count in each window. The
packet count is the number of packet arrival timestamps recorded in each
time window.  Finally, we normalized each packet count time series using min-max
normalization.

\begin{figure}[t]\centering
\includegraphics[width=\columnwidth]{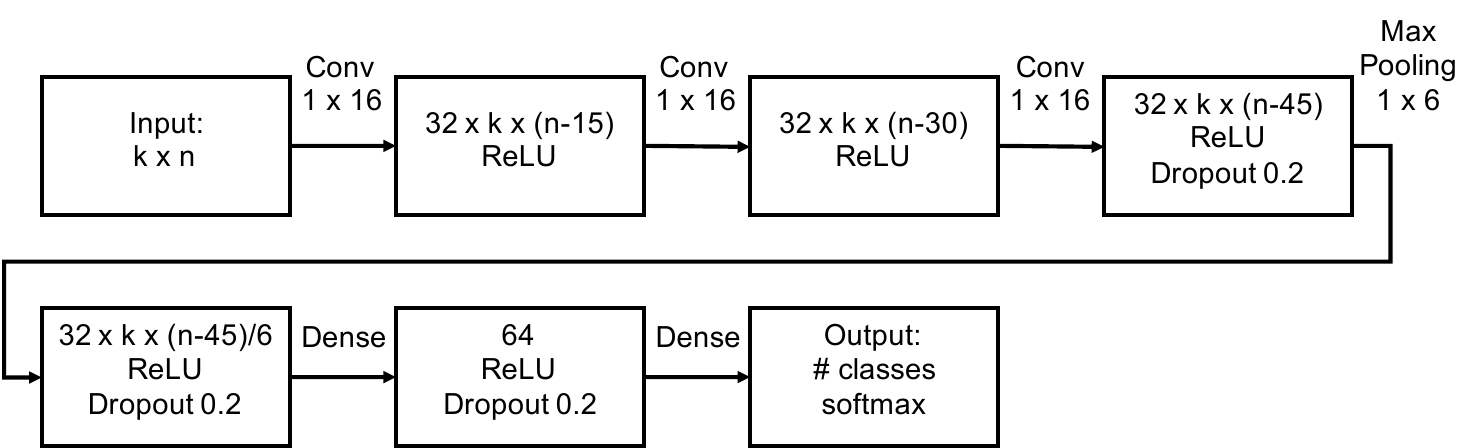}
\caption{CNN architecture. $k$: the number of features used.
  $n$: the number of elements of one time series, which is the total time of the
  time series divided by the window size (50ms).}
\label{fig:pacer:attack-cnn}
\end{figure}

Next, we implemented a CNN classifier to train on the time series of normalized packet
counts. Figure~\ref{fig:pacer:attack-cnn} shows the architecture of our
classifier, which consists of three convolution layers, a max pooling layer, and
two dense layers. We use a dropout of 0.2 between each pair of hidden layers of
the classifier as shown in the figure. We train the classifier with an Adam
optimizer \cite{kingma2014adam}, categorical cross-entropy error function, and
with input batches of 64 samples.
Our CNN classifier is similar to the one built by
Schuster \etal~\cite[section 7.2]{beautyburst}, with the difference that we used a
dropout of 0.2 between the model's hidden layers and 64 epochs for training.

We implemented the classifier using Tensorflow 2 API and with the Keras frontend.
We used 70\% of the time series data for each label (video) for training and the remaining for
evaluating the classifier.
The classifier achieves an overall precision and recall of 81.8\% each, and an
accuracy of 96.4\%.

\smallskip\noindent{\update{\bf Additional attack setups.}}
We performed a similar, but simpler attack on two additional setups:
(i) with the host, $S_1$'s NIC configured in SR-IOV, exposing vNICs to the
victim and attacker VMs, and
(ii) in a commercial IaaS Cloud provider platform.
In both cases, we were able to show that victim transfers of large files
generate a large signal on an attacker's cross-traffic that is visible in a
timeseries plot even to the naked eye. These attacks are not surprising, since
any queueing policy that allows a tenant to use network bandwidth not currently
used by other tenants that share a link permits {\nsc}s.

Our experiment confirms prior work~\cite{agarwal2016moving, song2001timing,
hintz2002, vila2017loophole, beautyburst, wright2008spot,
brumley2005remote, brumley2011remote, panchenko2011} and shows that
a network side-channel attack can identify videos in a collection with good
accuracy. While an attack in a production environment faces additional
challenges like achieving colocation with the victim, prior work has shown that it is easy to
attain colocation~\cite{getoffmycloud, inci2015seriously,
inci2016efficient}. Hence, cloud tenants that require strong confidentiality have
to consider that {\nsc}s are a realistic threat.

\section{\update{Security of masking mechanisms}}
\label{sec:eval-masking}



\update{Recall from \S\ref{sec:hypace} that {\sys} relies on four
  parameters whose values are empirically determined: the epoch
  length, the packet transmission batch size, HyPace's interrupt
  handler masking delay ($\delta_{xmit}$), and GPace's inbound packet-
  and timer-processing masking delay ($\delta_{delay}$). Of these,
  only the last two parameters are security-relevant. In this section,
  we discuss experiments to demonstrate that (i) masking is necessary
  (without it, the actual runtime of the handler tasks are observable,
  which could be correlated with guest secrets), and (ii) our
  empirically computed thresholds are effective at masking these
  timing leaks.  We do this by analyzing {\sys}'s handlers under two
  extreme configurations: no background workload ({\bf nobg}) and with
  {\em heavy} background load ({\bf bg}).  }


\begin{figure}[t]
\begin{minipage}{1.0\columnwidth}
	\includegraphics[width=\columnwidth]{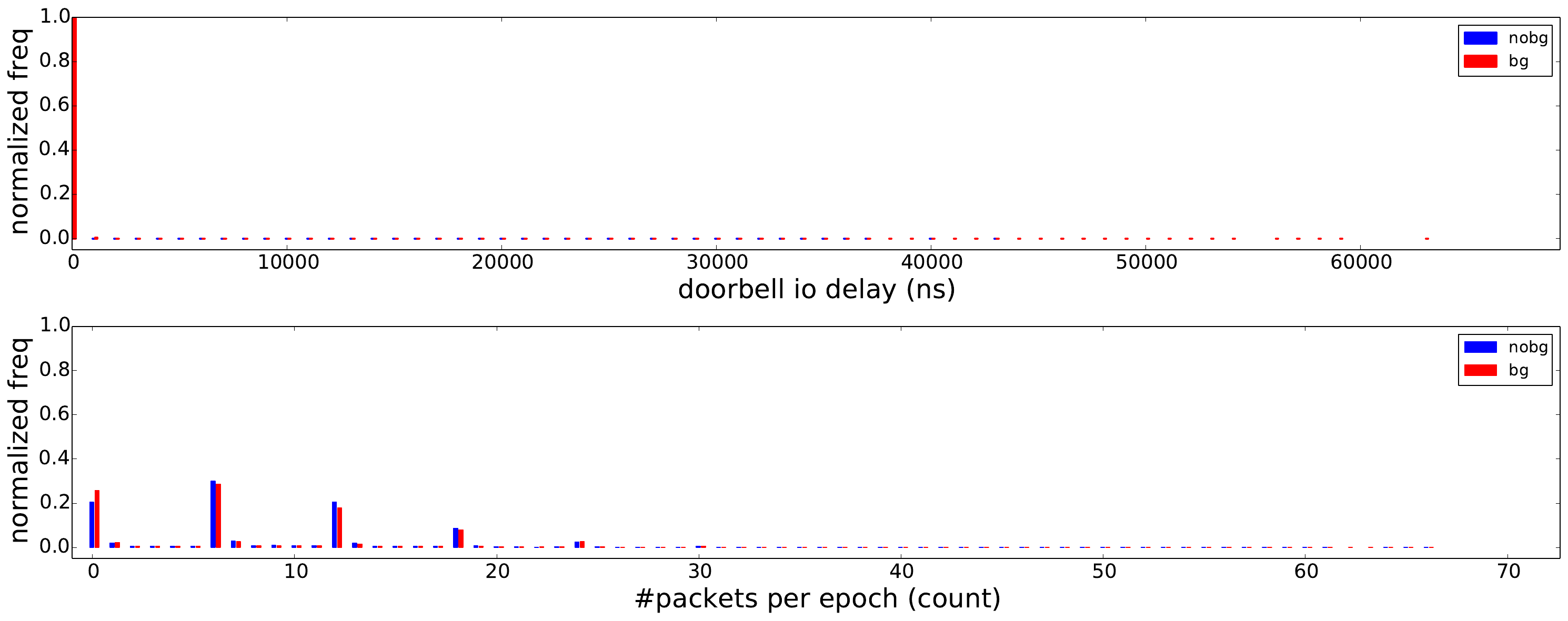}
	\centering
  \caption{\update{HyPace delays and batch size without masking.
  }}
  \label{fig:nomasking-bgload}
\end{minipage}
\begin{minipage}{1.0\columnwidth}
  \includegraphics[width=\columnwidth]{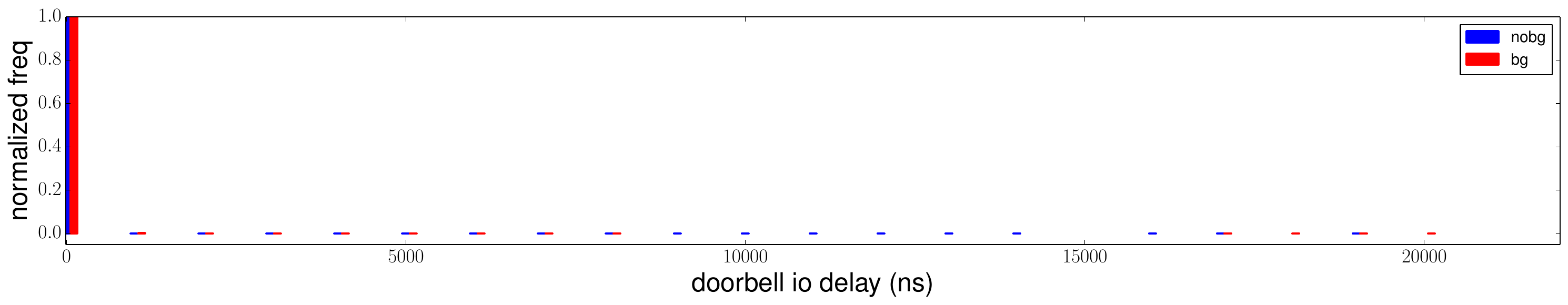}
  \centering
  \caption{\update{HyPace delays after applying masking.}}
  \label{fig:masking-bgload}
\end{minipage}
\end{figure}

To demonstrate that masking is necessary, we measure if there is any
difference in HyPace execution due to background load in the {\em
  absence of masking}.  Figure~\ref{fig:nomasking-bgload} top and
bottom plots respectively show the distributions of delays in HyPace's
doorbell writes (from the time of the scheduled interrupt handler) and
the number of packets that can be prepared within an 120$\mu$s epoch
in the two configurations. Both the distributions are based on 24-hour
experiments, with about 550 million epochs involving some packet
transmissions.
First of all, in absence of masking, the exact HyPace delay and batch size of
every single epoch is observable. Each pair of delay and batch size may be
correlated with specific secrets; thus the observations could leak the secrets.
Additionally, the distributions of delays and batch sizes are influenced by the
background workload. As the figure shows, the maximum HyPace delay observed (top
plot) is ~65$\mu$s and 44$\mu$s in {\bf bg} and {\bf nobg} configurations,
respectively and the maximum batch size (bottom plot) is 66 and 65 packets,
respectively. These observations show that the adversary can potentially affect
HyPace timing in the absence of masking to induce leaks. Hence, masking these
delays is essential for security.

Next, we repeat the experiment with masking enabled.
Figure~\ref{fig:masking-bgload} shows the observed delays of HyPace's
doorbell writes when masking with $\delta_{xmit}$ set to 35$\mu$s. As
can be seen,
\update{all handler execution times were masked in these
  experiments.}

For full disclosure, in earlier unrelated experiments, we had observed
a small number of epochs (\eg less than 20~out of 550 million) where
HyPace's delays exceeded $\delta_{xmit}$ by up to 5$\mu$s.
Such overruns, if they occurred in practice, would be mitigated
quickly by the automatic adjustment of $\delta_{xmit}$
(\S\ref{sec:hypace}).


To summarize, masking is necessary and effective. In the unlikely case
of HyPace delays exceeding our $\delta_{xmit}$ of 35$\mu$s, the
adversary would be able to observe execution times that may be
correlated with victim secrets.
However, \sys's automatic adjustment of $\delta_{xmit}$ denies the
adversary an opportunity to repeat an observation. Combined with the
adversary's challenge to induce increasing delays in the execution of
the privileged handlers, it seems impossible for an adversary to cause
a sufficient number of repeated overruns necessary to infer a victim's
secrets.

Our empirical observations and security arguments for $\delta_{delay}$ are
similar to the above, and we omit those details.

\section{Formal Model and Proof of Security}
\label{sec:formal-model}

\newcommand{\metafunc}[1]{\mathsf{#1}}
\newcommand{\kw}[1]{\mathtt{#1}}
\newcommand{\mytab}{~~~~~~~~}

\newcommand{\defeq}{\mathrel{\triangleq}}
\newcommand{\roundup}[1]{\lceil #1 \rceil}
\newcommand{\append}{\mathrel{+\!\!+}}
\newcommand{\merge}{\mathrel{\divideontimes}}

\newcommand{\nin}{\medskip\noindent}
\newcommand{\guest}{\mathsf{G}}
\newcommand{\env}{\mathsf{E}}
\newcommand{\hypace}{\mathsf{H}}
\newcommand{\state}{\sigma}
\newcommand{\sG}{{\state_{\guest}}}
\newcommand{\sE}{{\state_{\env}}}
\newcommand{\sH}{{\state_{\hypace}}}
\newcommand{\sGU}{{{\state}\mathsf{pub}_{\guest}}}
\newcommand{\sGR}{{{\state}\mathsf{pri}_{\guest}}}
\newcommand{\sEU}{{{\state}\mathsf{pub}_{\env}}}
\newcommand{\sER}{{{\state}\mathsf{pri}_{\env}}}
\newcommand{\prof}{\Phi}

\newcommand{\glob}{\mathsf{g}}
\newcommand{\sglob}{\state\glob}

\newcommand{\queue}{\mathsf{Q}}
\newcommand{\queueEG}{{\queue_{\guest}}}
\newcommand{\queueGH}{{\queue_{\hypace}}}
\newcommand{\queueGHu}{{\queue^u_{\hypace}}}
\newcommand{\queueGHp}{{\queue^p_{\hypace}}}
\newcommand{\queueHE}{{\queue_{\env}}}
\newcommand{\updates}{{\mathsf{U}}}
\newcommand{\updatesproj}{{\mathsf{V}}}

\newcommand{\ts}{\mathsf{T}}
\newcommand{\tinn}{\ts\mathsf{i}}
\newcommand{\tupd}{\ts\mathsf{u}}
\newcommand{\teff}{\ts\mathsf{e}}
\newcommand{\teffmax}{{\teff_{\max}}}
\newcommand{\tout}{\ts\mathsf{o}}
\newcommand{\tglob}{\ts\glob}

\newcommand{\event}{\mathsf{E}}
\newcommand{\evinn}{\event\mathsf{i}}
\newcommand{\nf}{\mathsf{f}}
\newcommand{\nfc}{\mathtt{flow}}

\newcommand{\evupd}{\event\mathsf{u}}
\newcommand{\evout}{\event\mathsf{o}}

\newcommand{\config}{\mathcal{C}}

\newcommand{\tdiff}{\delta}

\newcommand{\step}{\mathrel{\rightsquigarrow}}
\newcommand{\tepoch}{\tdiff_{e}}

\newcommand{\func}{\mathsf{F}}
\newcommand{\fE}{\func_\env}
\newcommand{\fG}{\func_\guest}
\newcommand{\fH}{\func_\hypace}
\newcommand{\fHu}{\func_\hypace^u}
\newcommand{\fHo}{\func_\hypace^o}

\newcommand{\bef}[1]{{\leq#1}}
\newcommand{\aft}[1]{{>#1}}
\newcommand{\res}[1]{|_{#1}}
\newcommand{\metaonly}{}

\newcommand{\resf}[2]{{[#1]_{#2}}}

\newcommand{\fupdprof}{\metafunc{update\_prof}}
\newcommand{\fupdate}{\metafunc{update}}

\newcommand{\rname}[1]{\mbox{#1}}

\newcommand{\invdelay}{{I_{\mathsf{delay}}}}
\newcommand{\inveffmax}{{I_{\mathsf{emax}}}}
\newcommand{\invupds}{{I_{\mathsf{upd}}}}

We build a formal model to prove formally that HyPace is secure. In
particular, it makes packet timing independent of any application
secrets. \update{The model assumes that the masking delays are never
exceeded.
}

Our model has the following actors:

\nin[-] Guest. The guest VM using HyPace's services.  Denoted
$\guest$. $\guest$ is modeled as a state machine with an internal
state $\sG$ that may contain secrets. Our goal is to keep these
secrets confidential from the other actors. $\guest$'s state machine
reacts to incoming network events and, in turn, generates events to
which HyPace reacts.

\nin[-] HyPace. Denoted $\hypace$. HyPace is modeled as a state
machine with internal state $\sH$. $\hypace$'s state machine reacts to
events generated by the guest $\guest$ and produces outgoing network
events to which the environment reacts.

\nin[-] The environment, which comprises everything outside the
server, including the network and the clients. Denoted $\env$. $\env$
is modeled as a state machine with internal state $\sE$. This state
machine reacts to HyPace's outgoing network events and produces
incoming network events to which the guest reacts.

For simplicity of exposition, we assume that there is a fixed set of
$N$ flow names, $\nfc_1, \ldots, \nfc_N$. Not all of these may be
active at any time, but we assume that HyPace always has a profile for
each of them. When a flow is not active, its (default) profile causes
no packets to be sent. We use $\nf$ and its decorated variants as
meta-variables that range over $\nfc_1, \ldots, \nfc_N$.

\paragraph{Event queues}
The three causal interactions $\env \rightarrow \guest$, $\guest
\rightarrow \hypace$, and $\hypace \rightarrow \env$ are mediated by
three event queues (producer-consumer queues), written $\queueEG$,
$\queueGH$ and $\queueHE$, respectively. A queue contains
\emph{timestamped pending events} that have been generated by the
queue's producer ($\env$ state machine for $\queueEG$), but not yet
handled by the queue's consumer ($\guest$ for $\queueEG$). The
subscript on a queue indicates its consumer. We describe these queues
next.

\nin $\queueEG$: This queue is a set of tuples of the form $(\tinn_i,
\evinn_i, \nf_i)$ meaning that incoming network event $\evinn_i$ that
occurred at time $\tinn_i$ on flow $\nf_i$ is still pending for the
guest. Incoming network events represent incoming packets (including
new client requests), network congestion signals and indicators of
packet loss (e.g., network stack timeouts). The exact structure of
these events is irrelevant here and, hence, kept abstract.
\[\begin{array}{lll}
\queueEG  & ::= &
(\tinn_1, \evinn_1, \nf_1), \ldots, (\tinn_n, \evinn_n, \nf_n)
\end{array}\]

\nin $\queueGH$: This queue actually consists of two subqueues -- the
\emph{profile update subqueue}, written $\queueGHu$, and the
\emph{packet subqueue}, written $\queueGHp$ -- and a time value
$\teffmax$.
\[ \begin{array}{lll}
\queueGH  & ::= & (\queueGHu, \teffmax, \queueGHp)
\end{array}\]

$\queueGHu$ contains profile updates. It is actually a key-value
store, keyed by flows. For every flow $\nfc_i$ ($i \in \{1\ldots N\}$),
the value is a set $\updates_i$ of pending updates on that flow.
\[\begin{array}{lll}
\updates & ::= & (\tupd_1, \evupd_1, \teff_1), \ldots, (\tupd_m, \evupd_m, \teff_m) \\
& & \mbox{{-}- sorted ascending by $\teff_i$}\\
\queueGHu & ::= & \nfc_1 \mapsto \updates_1, \ldots, \nfc_N \mapsto \updates_N 
\end{array}\]
$\updates$ contains update tuples of the form $(\tupd_i, \evupd_i,
\teff_i)$ meaning that the profile update described by $\evupd_i$ was
queued by the guest at time $\tupd_i$, but \emph{should be effective
  at time $\teff_i$}. $\evupd_i$ may replace the existing profile at
the application's request, or update/start the profile in response to
an incoming network event (e.g., starting a new profile in response to
a new request, pausing or resuming a profile in response to congestion
signals, or extending a profile in response to packet
retransmissions).

It is assumed that the guest chooses $\teff_i$ sufficiently after
$\tupd_i$ to allow the event to propagate to HyPace despite any
processing delays (we explicitly model this assumption later). It is
essential that \emph{$\teff_i$ be independent of secrets}. For an
update in response to an incoming network event, $\teff_i$ can be set
to the timestamp of the incoming network event plus the maximum
(empirical) propagation delay of the network stack. For a profile
update initiated by the application, the application can set $\teff_i$
to the end time of the current profile minus the propagation delay of
the application and network stack. If the propagation delay of the
application is secret-dependent, it can be bucketized into public
buckets and the next higher bucket boundary can be used instead.

Note that $\updates$ is an ordered list, not a set. It is sorted in
increasing order of $\teff_i$. HyPace also applies updates in this
order. We often treat $\queueGHu$ as an indexed vector, writing
$\queueGHu[\nf]$ for the update tuples of flow $\nf$. 

$\queueGH$ contains a timestamp $\teffmax$, which is the
highest effective time of any update event that has been added to
$\queueGHu$ in the past. This means that if $(\tupd, \evupd, \teff)
\in \updates_i$, then $\teff \leq \teffmax$.

The packet subqueue $\queueGHp$ contains network packets queued by the
guest for transmission. These packets should be encrypted. The details
of this queue are irrelevant for our model, so we leave it abstract.

\nin $\queueHE$: This queue is a set of tuples of the form $(\tout_i,
\evout_i, \nf_i)$ meaning that the outgoing network packet $\evout_i$
is generated on flow $\nf_i$ by HyPace at time $\tout_i$. The
$\evout_i$ represents the encrypted payload, so its structure is
irrelevant.
\[\begin{array}{lll}
\queueHE  & ::= &
(\tout_1, \evout_1, \nf_1), \ldots, (\tout_n, \evout_n, \nf_n)
\end{array}\]

\paragraph{State}
Next, we describe the internal states of the environment, guest and
HyPace:

\nin $\sE$: The environment's state $\sE$ is a pair of a private
component and a remaining component, written $\sER$ and $\sEU$,
respectively. The private component $\sER$ is held only by the clients
of the guest VM being protected, while $\sEU$ represents the remaining
state of the clients, the network and any other actors. We do not
specify these any further, but there are constraints on how they
evolve.
\[\begin{array}{lll}
\sE  & ::= & (\sER, \sEU)
\end{array}\]

\nin $\sG$: The guest's state $\sG$ is similarly a pair of a private
component and a remaining component, written $\sGR$ and $\sGU$,
respectively. We do not specify these any further, but there are
constraints on how they evolve.
\[\begin{array}{lll}
\sG  & ::= & (\sGR, \sGU)
\end{array}\]

\nin $\sH$: HyPace's state $\sH$ consists of a map from flows to the
current active profiles on them. We use the notation $\prof$ for a
profile. The exact structure of profiles is irrelevant for the
security argument, so we keep it abstract.
\[\begin{array}{lll}
\sH  & ::= & \nfc_1 \mapsto \prof_1, \ldots, \nfc_N \mapsto \prof_N
\end{array}\]

\paragraph{Auxiliary state}
There is also some auxiliary state that is not associated to any
specific component. This state is just the current (global) time,
written $\tglob$.
\[\begin{array}{lll}
\tglob & ::= & \mbox{Current global time} 
\end{array}\]
%

\paragraph{Overall state (Configuration)}
The overall state of the system, also called a \emph{configuration}
and denoted $\config$, consists of the internal states of HyPace, the
guest and the environment, the three event queues, and the auxiliary
state.
\[\begin{array}{lll}
\config  & ::= & (\sE, \sG, \sH, \queueEG, \queueGH, \queueHE, \tglob)
\end{array}\]


\subsection{System evolution}

The overall state (configuration) evolves over time through
transitions. We write $\config \step \config'$ to say that the
configuration $\config$ transitions to $\config'$ in a single step.

A transition happens when one of the agents acts on events pending for
it. Without loss of generality, we assume that only HyPace's actions
cause the global time $\tglob$ to jump forward, and this jump is
exactly the length of one epoch, which we denote $\tepoch$.
\[ \begin{array}{lll}
\tepoch & ::= & \mbox{Length of epoch} 
\end{array}\]
Also w.l.o.g., the agents act in order: HyPace, environment, guest,
and repeat. Technically,
\[\step ~\defeq~ \step_\env ~;~ \step_\guest ~;~ \step_\hypace\]
where $\step_\env$, $\step_\guest$ and $\step_\hypace$ represent steps
of the environment, guest and HyPace, respectively, and the semicolon
means relation composition. In the following we describe $\step_\env$,
$\step_\guest$ and $\step_\hypace$ one by one.


\begin{figure}
\[\begin{array}{lcl}
\queueHE\res{\metaonly} & \defeq & \{ (\tout, \nf) ~|~ (\tout, \evout, \nf) \in \queueHE \}
\\
\queueHE_1 \sim \queueHE_2 & \defeq & \queueHE_1\res{\metaonly} = \queueHE_2\res{\metaonly}
\\\\
\queueEG\res{\metaonly} & \defeq & \{ (\tinn, \nf) ~|~ (\tinn, \evinn, \nf) \in \queueEG \}
\\
\queueEG_1 \sim \queueEG_2 & \defeq & \queueEG_1\res{\metaonly} = \queueEG_2\res{\metaonly}
\\\\
\updates\res{\metaonly} & \defeq & [ (\evupd_i, \teff_i) ~|~ (\tupd_i, \evupd_i, \teff_i) \in \updates ]
\\
\queueGHu\res{\metaonly} & \defeq & \{ (\nf \mapsto \updates\res{\metaonly}) ~|~ (\nf \mapsto \updates) \in \queueGHu \}
\\
\queueGHu_1 \sim \queueGHu_2 & \defeq & \queueGHu_1\res{\metaonly} = \queueGHu_2\res{\metaonly}
\\\\
\sE\res{\metaonly} & = & \sEU \mbox{ for } \sE = (\sER, \sEU)
\\
\sE_1 \sim \sE_2 & \defeq & \sE_1 \res{\metaonly} = \sE_2 \res{\metaonly}
\\\\
\sG\res{\metaonly} & = & \sGU \mbox{ for } \sG = (\sGR, \sGU)
\\
\sG_1 \sim \sG_2 & \defeq & \sG_1 \res{\metaonly} = \sG_2 \res{\metaonly}
\end{array}
\]

  \caption{Equivalence of states and queues}
  \label{fig:equivalence}
\end{figure}

\begin{figure}
  \framebox{Assumptions}
  \begin{mathpar}
    \begin{array}{ll}
      (1) & 
      \begin{array}[t]{l}
        \queueHE_1 \sim \queueHE_2 \mbox{ and } \sE_1 = \sE_2 \mbox{ and } \\
        \fE(\queueHE_1, \sE_1, \tglob) = (\queueHE_1', \queueEG_1'', \sE_1') \mbox{ and } \\
        \fE(\queueHE_2, \sE_2, \tglob) = (\queueHE_2', \queueEG_2'', \sE_2') \\
        \Rightarrow \\
        ~~~~ \queueHE_1' \sim \queueHE_2' \mbox { and } \\
        ~~~~ \queueEG_1'' \sim \queueEG_2'' \mbox { and } \\
        ~~~~ \sE_1' \sim \sE_2'
      \end{array}
    \end{array}
  \end{mathpar}
  \framebox{Transition}
  \begin{mathpar}
    \inferrule{
      \fE(\queueHE, \sE, \tglob) = (\queueHE', \queueEG'', \sE')
    }{
      (\sE, \sG, \sH, \queueEG, \queueGH, \queueHE, \tglob) \\\\
      \step_\env
      (\sE', \sG, \sH, \queueEG \cup \queueEG'', \queueGH, \queueHE', \tglob)
    }\rname{env}
  \end{mathpar}
  \caption{Assumptions and transition of the environment}
  \label{fig:env}
\end{figure}

\subsubsection{Environment acts}
The environment acts by consuming a subset of events in the queue
$\queueHE$, processing them to update its internal state and adding
new events to the queue $\queueEG$. We model the environment as an
abstract function $\fE$ that takes as input the current queue
$\queueHE$, the environment's internal state $\sE$ (which contains the
private component $\sER$ and the remaining component $\sEU$) and the
current global time $\tglob$. It outputs a new internal state $\sE'$,
an updated queue $\queueHE'$ (which should be a subset of $\queueHE$)
and a set of events $\queueEG''$ that are added to $\queueEG$ by
$\step$.
\[ \fE(\queueHE, \sE, \tglob) = (\queueHE', \queueEG'', \sE') \]

Using $\fE$, we define the transition rule ($\rname{env}$) for the
environment (Figure~\ref{fig:env}). Note the index $\env$ in
$\step_\env$, which indicates that this is the environment's
transition.\footnote{The way to read a rule $\inferrule{A}{B}$ is that
  \emph{if} $A$ holds \emph{then} $B$ holds.}

The function $\fE$ can be arbitrary (we don't assume that we know what
the network and the clients do), but it is subject to an important
security assumption, which we describe here.

\nin (Clients don't break secrecy explicitly) Clients get to see the
payloads of incoming messages, which may depend on secrets and also
have access to their existing private state $\sER$. In our threat
model, we explicitly trust clients to not leak either of these into
timing. In the formal model, this is specified by a constraint on
$\fE$. Specifically, if we consider two input queues $\queueHE_1$ and
$\queueHE_2$ that differ only in payloads (but agree on timing), and
two states $\sE_1$ and $\sE_2$ that differ only in the private
components $\sER_1$ and $\sER_2$ (but agree in the non-private
components), then the output queues $\queueHE_1'$ and $\queueHE_2'$
should differ only in the payloads (similarly for $\queueEG''_1$ and
$\queueEG_2''$), and the output states $\sE_1$ and $\sE_2$ can differ
only in the private components.

Formally, we define $\queueHE_1 \sim \queueHE_2$ to mean that
$\queueHE_1$ and $\queueHE_2$ agree on the flows and timestamps of
events. Similarly, we define $\queueEG_1 \sim \queueEG_2$. Finally, we
define $\sE_1 \sim \sE_2$ to mean that $\sE_1$ and $\sE_2$ agree on
the non-private components. These definitions are shown in
Figure~\ref{fig:equivalence}. We then make the assumption (1) in
Figure~\ref{fig:env}, which captures exactly the intuition described
in the previous paragraph.

\paragraph{Note.}
A real $\fE$ would also have the following properties, but we do not
need these properties for security, so we do \emph{not} assume
them. We show these properties just for completeness.

\nin (Causality) $\fE$ should depend only on past events in
$\queueHE$, i.e., those that occurred before the current time
$\tglob$. Formally, we assume that
\[ \fE(\queueHE, \sE, \tglob) = \fE(\queueHE\res{\bef{\tglob}}, \sE, \tglob) \]
Here, $\queueHE\res{\bef{\tglob}}$ denotes the subset of $\queueHE$
containing events whose timestamps are no more than $\tglob$.
\[ \queueHE\res{\bef{\tglob}} \defeq \{ (\tout, \evout, \nf) ~|~  (\tout, \evout, \nf) \in \queueHE \mbox{ and } \tout \leq \tglob \} \]

\nin (Non-modification of past outputs) $\fE$ should not output events
in the past, i.e., $\queueEG''$ should not contain any events with
timestamps $\tglob$ or lower.
\[ \forall (\tinn, \evinn, \nf) \in \queueEG''. \; \tinn > \tglob \]

\nin (Non-consumption of future inputs) $\fE$ should not consume input
events from the future, i.e., $\queueHE'$ should agree with $\queueHE$
on \emph{future} events, i.e.,
\[ \queueHE\res{\aft{\tglob}} = \queueHE'\res{\aft{\tglob}} \]
Here, $\queueHE\res{\aft{\tglob}}$ denotes the subset of $\queueHE$
containing events whose timestamps are strictly greater than
$\tglob$. It is defined analogous to $\queueHE\res{\bef{\tglob}}$.


\begin{figure}
  \framebox{Assumptions}
  \begin{mathpar}
    \begin{array}{ll}
    (2) & 
      \begin{array}[t]{l}
        \queueEG_1 \sim \queueEG_2 \mbox{ and } \sG_1 \sim \sG_2 \mbox{ and }\\
        \fG(\queueEG_1, \teffmax_1, \queueGHp_1, \sG_1, \tglob) = \\
        ~~~~~(\queueEG_1', \queueGHu_1'', \teffmax_1', \queueGHp_1', \sG_1') \mbox{ and } \\
        \fG(\queueEG_2, \teffmax_2, \queueGHp_2, \sG_2, \tglob) =\\
        ~~~~~(\queueEG_2', \queueGHu_2'', \teffmax_2', \queueGHp_2', \sG_2') \\
        \Rightarrow \\
        %
        ~~~~ \queueEG_1' \sim \queueEG_2' \mbox { and } \\
        ~~~~ \queueGHu_1'' \sim \queueGHu_2'' \mbox { and } \\
        ~~~~ \sG_1' \sim \sG_2'
      \end{array}
      \\\\
      (3) &
      \begin{array}[t]{l}
        \fG(\queueEG, \teffmax, \queueGHp, \sG, \tglob) = \\
        ~~~~(\queueEG', \queueGHu', \teffmax', \queueGHp', \sG') \\
        \Rightarrow
        \invdelay(\queueGHu')
      \end{array}
      \\
      \multicolumn{2}{l}{
        \begin{array}{l@{}l}
          \mbox{where} \\
          \invdelay(\queueGHu) ~\defeq~ & \forall (\nf \mapsto \updates) \in \queueGHu.\;
          \forall (\tupd_i, \evupd_i, \teff_i) \in \updates.
          \\
          & ~~~~~~~~~~ \tupd_i \leq \teff_i
        \end{array}
      }
      \\\\
      (4) &
      \begin{array}[t]{l}
        \fG(\queueEG, \teffmax, \queueGHp, \sG, \tglob) = \\
        ~~~~ (\queueEG', \queueGHu'', \teffmax', \queueGHp', \sG') \\
        \Rightarrow
        \inveffmax(\queueGHu'', \teffmax')
      \end{array}
      \\
      \multicolumn{2}{l}{
        \begin{array}{l@{}l}
          \mbox{where} \\
          \inveffmax(\queueGHu, \teffmax) ~\defeq~ & \forall (\nf \mapsto \updates) \in \queueGHu.\;
          \forall (\tupd_i, \evupd_i, \teff_i) \in \updates.
          \\
          & ~~~~~~~~~~ \teff_i \leq \teffmax
        \end{array}
      }
      \\\\
      (5) &
      \begin{array}[t]{l}
        \fG(\queueEG, \teffmax, \queueGHp, \sG, \tglob) = \\
        ~~~~~ (\queueEG', \queueGHu'', \teffmax', \queueGHp', \sG') \\
        \Rightarrow
        \forall (\nf \mapsto \updates') \in \queueGHu''. \;
        \forall (\tupd_i', \evupd_i', \teff_i') \in \updates'. \\
        ~~~~~~~~~~~\teffmax < \teff_i'
      \end{array}
    \end{array}
  \end{mathpar}
  \framebox{Transition}
  \begin{mathpar}
    \inferrule{
      \queueGH = (\queueGHu, \teffmax, \queueGHp) \\ 
      \fG(\queueEG, \queueGHp, \sG, \tglob) = (\queueEG', \queueGHu'', \teffmax', \queueGHp', \sG') \\
      \queueGH' \leftarrow (\queueGHu \merge \queueGHu'', \teffmax', \queueGHp')
    }{
      (\sE, \sG, \sH, \queueEG, \queueGH, \queueHE, \tglob) \\\\
      \step_\guest
      (\sE, \sG', \sH, \queueEG', \queueGH', \queueHE, \tglob)
    }\rname{guest}
  \end{mathpar}
  Note: $\merge$ is the merge operation on (sorted) lists, lifted
  pointwise to key-value tuples pointwise on keys.\\
  \caption{Assumptions and transition of the guest}
  \label{fig:guest}
\end{figure}

\subsubsection{Guest acts}
The guest acts by consuming events from $\queueEG$ to update its
internal state and to produce events in the queue $\queueGH$
(including both its subqueues $\queueGHu$ and $\queueGHp$). We model
the environment as an abstract function $\fG$.
\[ \fG(\queueEG, \teffmax, \queueGHp, \sG, \tglob) = (\queueEG', \queueGHu'', \teffmax', \queueGHp', \sG') \]
$\fG$ takes as input the incoming network event queue $\queueEG$
(recall that this queue is populated by the environment), the current
maximum update effective time $\teffmax$, the current packet queue
$\queueGHp$, the guest's current state $\sG$ and the current time
$\tglob$. It outputs an updated input queue $\queueEG'$, a set of
profile update key-values $\queueGHu''$ to add to the hypervisor's
update queue, a new $\teffmax'$, an updated packet queue $\queueGHu'$,
and a new guest state $\sG'$.

Using $\fG$, we define the transition rule ($\rname{guest}$) for the
guest (Figure~\ref{fig:guest}). The index $\guest$ in $\step_\guest$
indicates that this is the guest's transition.

The function $\fG$ can be arbitrary (meaning that the enforcement is
almost black-box), but it is subject to some causality and security
assumptions.

\nin (Guest does not break secrecy explicitly) The guest should
distinguish private from public state, but it may have timing
leaks. Specifically, the descriptions of profile updates (denoted
$\evupd_i$) it queues for the hypervisor must not depend on the
guest's secret state or the payloads of incoming packets, which may
also be secret-dependent. Similarly, the times at which these updates
become effective (determined by $\teff_i$) should be
secret-independent. However, the time at which the update events are
queued (denoted $\tupd_i$) may depend on secrets due to timing
leaks. Also, the packets the guest queues to send (i.e., the subqueue
$\queueGHp$) may depend on secrets. These packets are encrypted
anyhow.

To formalize this, we define notions of equivalence $\sim$ of the
guest state $\sG$ (the public components, but not the private
components, must coincide) and the update event queue $\queueGHu$
(Figure~\ref{fig:equivalence}). We then assume (2) from
Figure~\ref{fig:guest}, which formalizes the intuition of the previous
paragraph.
Observe that $\queueGHu_1'' \sim \queueGHu_2''$ in (2) correctly
imposes no restrictions on $\tupd_i$s as these may depend on
secrets. Additionally, (2) imposes no restrictions at all on the
packet queues ($\queueGHp_1'$ and $\queueGHp_2'$), as these queues may
also be secret-dependent.

\nin (Propagation delays are respected) Next, we formalize the fiat
assumption that the guest accounts for propagation delays
correctly. For this, we define a property $\invdelay(\queueGHu)$ on
profile update subqueues, and assume that this property holds of the
output $\queueGHu''$ of $\fG$ (assumption (3) in
Figure~\ref{fig:guest}). $\invdelay(\queueGHu)$ simply says that for
any update tuple $(\tupd_i, \evupd_i, \teff_i)$ in $\queueGHu$,
$\tupd_i \leq \teff_i$, meaning that the time at which the update gets
queued ($\tupd_i$) is no more than the intended effective time
$\teff_i$.

\nin (Guest queues updates in increasing order) The new updates the
guest queues ($\queueGHu''$) should have effective times after
$\teffmax$, and the new $\teffmax'$ should be an upper bound on the
effective times in $\queueGHu'$. We formalize these as assumptions~(4)
and~(5) in Figure~\ref{fig:guest}.

\paragraph{Note.}
Any real guest would also have the following additional
properties. These properties are not necessary for security, so we do
not assume them. We mention them just for completeness.

\nin (Causality) $\fG$ should only depend on past events in
$\queueEG$, i.e., those that occurred before $\tglob$.

\nin (Non-modification of past outputs) $\fG$ should not output events
in the past, i.e., $\queueGHu'$ should only contain events with
timestamps greater than $\tglob$.

\nin (Non-consumption of future inputs) $\fG$ should not remove future
events from its input queue, $\queueEG$. Formally, $\queueEG$ and
$\queueEG'$ should agree on events that have timestamps greater than
$\tglob$.


\begin{figure}
  \[\begin{array}{l}
  \hline\\[-10pt]
  \kw{function}~~ \fupdate(\prof, \updatesproj)
  \\ \hline \\[-6pt]
  \prof_{out} \leftarrow \prof \\
  \kw{foreach}~(\evupd_i, \teff_i) \in \updatesproj: \\
  \mytab\prof_{out} \leftarrow \fHu(\prof, \evupd_i, \teff_i)\\
  \kw{return}~\prof_{out}
  \\\hline
  \\\\
  \hline\\[-10pt]
  \kw{function}~~ \fupdprof(\prof, \updates, \tglob)
  \\ \hline \\[-6pt]
  i \leftarrow \min_j \{\updates[j] = (\tupd, \_, \_) \mbox{ and } \tupd > \tglob\} \\
  \updates_{curr} \leftarrow \updates[..(i-1)] \\
  \updates_{rest} \leftarrow \updates[i..] \\
  \prof_{out} \leftarrow \fupdate(\prof, \updates_{curr}\res{\metaonly}) \\
  \kw{return}~ (\prof_{out}, \updates_{rest})
  \\\hline
  \end{array}\]
  \caption{The functions $\fupdate$ and $\fupdprof$ that model
    HyPace's profile update logic}
  \label{fig:updprof}
\end{figure}


\begin{figure}
  \framebox{Assumptions}
  \begin{mathpar}
    \begin{array}{ll}
      (6) & 
      \begin{array}[t]{l}
        \fHu(\prof, \evupd, \teff) = \prof' \mbox{ and }
        \teff' < \teff \\
        ~~\Rightarrow
        \resf{\prof}{\teff'} = \resf{\prof'}{\teff'}
      \end{array}
      \\\\
      (7) &
      \begin{array}[t]{l}
        s_1 = s_2 \mbox{ and }\\
        \fHo(\queueGHp_1, s_1) = (\queueGHp_1', \queueHE_1'') \mbox{ and }\\
        \fHo(\queueGHp_2, s_2) = (\queueGHp_2', \queueHE_2'') \\
        \Rightarrow \\
        ~~~~ \queueHE_1'' \sim \queueHE_2''
      \end{array}
    \end{array}
  \end{mathpar}
  \framebox{Transition}
  \begin{mathpar}
    \inferrule{
      \sH = \{\nfc_i \mapsto \prof_i\}_{i=1}^N \\
      \queueGHu = \{\nfc_i \mapsto \updates_i\}_{i=1}^N \\
      (\prof_i', \updates_i') = \fupdprof(\prof_i,\updates_i,\tglob)\\
      \sH' \leftarrow \{\nfc_i \mapsto \prof_i'\}_{i=1}^N \\
      \resf{\sH'}{\tglob} \defeq \{\nfc_i \mapsto \resf{\prof_i'}{\tglob}\}_{i=1}^N \\
      \fHo(\queueGHp, \resf{\sH'}{\tglob}) = (\queueGHp', \queueHE'') \\
      \queueGHu' \leftarrow \{\nfc_i \mapsto \updates_i'\}_{i=1}^N     
    }{
      (\sE, \sG, \sH, \queueEG, \queueGH, \queueHE, \tglob) \\\\
      \step_\hypace
      (\sE, \sG, \sH', \queueEG, \queueGH', \queueHE \cup \queueHE'', \tglob + \tepoch)
    }\rname{hypace}
  \end{mathpar}
  \caption{Assumptions and transition of HyPace}
  \label{fig:hypace}
\end{figure}

\subsubsection{HyPace acts}

In its turn to act, HyPace's work corresponds to its (batched) actions
in the epoch $(\tglob, \tglob + \tepoch]$, where $\tepoch$ is the
  epoch length. HyPace does two things.

First, HyPace applies the longest prefix of profile updates from
$\queueGHu$ whose availability timestamps $\tupd_i$ (\emph{not}
effective timestamps $\teff_i$) are less than $\tepoch$. This models
applying pending updates that became available in the earlier
epochs. We assume an abstract profile update function $\fHu(\prof,
\evupd, \teff) = \prof'$ that updates a current profile $\prof$ to a
new profile $\prof'$ by applying the update $\evupd$ effectively from
$\teff$. Importantly, we assume that $\prof$ and $\prof'$ agree in the
packet timing they provide up to time $\teff$. In other words, the
update becomes effective only at time $\teff$. The updated profile
$\prof'$ is stored back in HyPace's internal state $\sH$.

In detail, for each flow $\nfc_i$, we iterate on the updates in
$\queueGHu[\nfc_i]$ from the left, till the first update event whose
availability timestamp $\tupd_j$ is larger than $\tglob$. All updates
to the left of this event came to HyPace before the end of the
previous epoch and are applied immediately to $\prof_i$ using the
function $\fHu$ and stored back in $\sH$. The remaining updates stay
pending in $\queueGHu[\nfc_i]$. This iteration is formalized by the
defined function $\fupdprof$ shown in Figure~\ref{fig:updprof}. The
function $\fupdate$.

Note that the meta-variable $\updatesproj$ stands for updates
\emph{projected} to only the update event $\evupd$ and the effective
time $\teff$ (i.e., removing $\tupd$).
\[ \begin{array}{lll}
  \updatesproj & ::= & (\evupd_1, \teff_1), \ldots, (\evupd_m, \teff_m)
\end{array} \]

Second, HyPace uses the updated profiles to generate output packets
for the NIC. Only profile prefixes up to $\tglob$ are considered.
This process is abstractly modeled by a function $\fHo(\queueGHp,
\resf{\sH}{\tglob}) = (\queueGHp', \queueHE'')$. This function takes
as input the current packet subqueue between the guest and HyPace, and
the current profiles on all flows, limited to the time interval $(0,
\tglob]$ (denoted by $\resf{\sH}{\tglob}$). It returns an updated
  packet subqueue $\queueGHp'$ (as some packets have been consumed),
  and a set of events $\queueHE''$ to output to the NIC.

The entire HyPace transition is formalized in the rule
($\rname{hypace}$) in Figure~\ref{fig:hypace}. The index $\hypace$ in
$\step_\hypace$ indicates that this is HyPace's transition.

We make the following assumptions about the functions $\fHu$ and
$\fHo$. 

(Profile update respects effective time) $\fHu$ should respect the
effective time $\teff_i$ passed as its third argument. Formally, we
let $\resf{\prof}{\ts}$ denote the restriction of profile $\prof$ to
the interval $(0, \ts]$. Assumption~(6) of Figure~\ref{fig:hypace}
  represents this requirement. (Note that $\resf{\prof}{\ts}$ is
  abstract; we don't define it. However, it is also used in the rule
  ($\rname{hypace}$) of Figure~\ref{fig:hypace} to limit profiles
  before using them to generate events.)

(HyPace does not leak secrets explicitly) $\fHo$ should not leak
  information from the packet subqueue $\queueGHp$, which may be
  secret-dependent, into the timing of outgoing packets. Formally,
  this is represented by assumption~(7) of
  Figure~\ref{fig:hypace}.

\subsection{Security theorem and proof}

\begin{figure}
  \begin{mathpar}
    \begin{array}{ll}
      \multicolumn{2}{l}{
        \framebox{Invariants (unary)}
        }\\[5pt]
      (I1) & \invdelay(\queueGHu) \\
      (I2) & \inveffmax(\queueGHu, \teffmax) 
      \\[5pt]
      \multicolumn{2}{l}{
        \mbox{where $\invdelay$ and $\inveffmax$ are defined in Figure~\ref{fig:guest}.}
      }
      \\[5pt]
      \multicolumn{2}{l}{
        \framebox{Invariants (relational)}
      }\\[5pt]
      (I3) &
      \begin{array}[t]{@{}l}
        \sE_1 \sim \sE_2 \mbox{ and} \\
        \sG_1 \sim \sG_2 \mbox{ and} \\
        \queueEG_1 \sim \queueEG_2 \mbox{ and} \\
        \invupds(\queueGHu_1, \sH_1, \queueGHu_2, \sH_2) \mbox{ and} \\
        \queueHE_1 \sim \queueHE_2 \mbox{ and} \\
        \tglob_1 = \tglob_2 
      \end{array}
      \\[5pt]
      \multicolumn{2}{l}{\mbox{where}} \\
      \multicolumn{2}{l}{
        ~~\begin{array}{@{}l}
          \invupds(\queueGHu_1, \sH_1, \queueGHu_2, \sH_2) ~\defeq\\
          ~~~~\forall i.\; (\nfc_i \mapsto \updates_1) \in \queueGHu_1 \mbox{ and } \\
          ~~~~~~~~~~ (\nfc_i \mapsto \updates_2) \in \queueGHu_2 \mbox{ and } \\
          ~~~~~~~~~~ (\nfc_i \mapsto \prof_1) \in \sH_1 \mbox{ and } \\
          ~~~~~~~~~~  (\nfc_i \mapsto \prof_2) \in \sH_2              \\
          ~~~~~~~~ \Rightarrow \\
          ~~~~~~~~~~~ \exists \updatesproj.\;
          (\updates_1\res{\metaonly} = \updatesproj \append \updates_2\res{\metaonly}
          \mbox{ and }
          \prof_2 = \fupdate(\prof_1, \updatesproj)) \mbox{ or} \\
          ~~~~~~~~~~~~~~~~~~
          (\updates_2\res{\metaonly} = \updatesproj \append \updates_1\res{\metaonly}
          \mbox{ and }
          \prof_1 = \fupdate(\prof_2, \updatesproj))
        \end{array}
      }
    \end{array}
  \end{mathpar}
  \caption{Invariants of the transition system}
  \label{fig:invariants}
\end{figure}


We formalize confidentiality of the guest's private state using the
standard concept of
\emph{noninterference}~\cite{DBLP:series/ais/Smith07}. Noninterference
is inherently a relational property, i.e., a property of two runs of
the system. Noninterference is usually proved by establishing
\emph{invariants}. We state and prove the relevant invariants of our
model before stating and proving security. Our model has two kinds of
relevant invariants: unary and relational.

\paragraph{Unary invariants}
A unary invariant is a property of the configuration that is
\emph{preserved} by all transitions, i.e., if the property holds
before the a transition, then it holds after the transition as
well. There are two unary invariants of relevance to us. These are
called $(I1)$ and $(I2)$, and shown in Figure~\ref{fig:invariants}.

\begin{lem}
$(I1)$ and $(I2)$ are (unary) invariants.
\end{lem}
\begin{proof}
  $(I1)$: We need to prove that every transition preserves $(I1)$,
  i.e., $\invdelay(\queueGHu)$. This property is defined pointwise on
  the individual elements of $\queueGHu$, so it can be violated only
  by a transition that \emph{adds} to $\queueGHu$. The only such
  transition is $\step_\guest$ ($\step_\env$ does not change
  $\queueGHu$ and $\step_\hypace$ removes from $\queueGHu$). However,
  $\step_\guest$ trivially preserves the invariant due to assumption~(3)
  of Figure~\ref{fig:guest}.

  $(I2)$: We need to prove that every transition preserves $(I2)$,
  i.e., $\inveffmax(\queueGHu, \teffmax)$. Again, this property is
  defined pointwise on the individual elements of $\queueGHu$, so we
  only need to consider the transition $\step_\guest$. This transition
  trivial preserves the invariant due to assumption~(4) of
  Figure~\ref{fig:guest}.
\end{proof}
  
\paragraph{Relational invariants}
A relational property is a property of two configurations,
conventionally denoted by the subscripts $1$ and $2$. A relational
property is called a \emph{relational invariant} if the following
holds: Consider two configurations in the property. If both
configurations step with the same kind of transition then the
resulting configurations are also in the property.\footnote{We do not
  need to consider different types of transitions on the two sides
  since we fix the order of the transitions. In other words, we assume
  a deterministic scheduler. This can be easily relaxed to any
  scheduler that only looks at the non-private components of the
  configurations.}

For our model, there is only one interesting relational invariant --
$(I3)$ of Figure~\ref{fig:invariants}. This relation says two things:
(1) The two configurations agree on the public (non-private)
components of $\sE$, $\sG$, $\queueEG$, and $\queueHE$ (the private
components may arbitrarily differ) and (2) For every flow $\nfc_i$,
HyPace's internal state $\sH$ and the pending updates queue
$\queueGHu$ differ across the two configurations only in that, in one
of the two sides, fewer updates have been taken out of $\queueGHu$ and
applied to the flow's profile. In other words, the \emph{same} profile
updates reach HyPace (and with the same effectiveness timestamps
$\teff_i$) in the two runs, but the two runs may differ in \emph{when}
they apply the updates. The latter difference arises because the time
at which the updates reach HyPace (the timestamps $\tupd_i$s) may
depend on guest secrets and may differ.

\begin{lem}
$(I3)$ is a relational invariant.
\end{lem}
\begin{proof}
We assume that the unary invariants $(I1)$ and $(I2)$ hold. We then
assume that $(I3)$ holds \emph{before} a step, and show that it holds
\emph{after} the step as well. To prove the latter, we show that all
conjuncts of $(I3)$ hold. We use the quote symbol $'$ to denote
elements after the transition.

\underline{$\sE_1 \sim \sE_2$}: The only transition that modifies $\sE$ is
$\step_\env$. This transition trivially guarantees $\sE_1' \sim
\sE_2'$ due to assumption~(1) of Figure~\ref{fig:env}.

\underline{$\sG_1 \sim \sG_2$}: The only transition that modifies $\sG$ is
$\step_\guest$. This transition trivially guarantees $\sG_1' \sim
\sG_2'$ due to assumption~(2) of Figure~\ref{fig:guest}.

\underline{$\queueEG_1 \sim \queueEG_2$}: The two transitions that
modify $\queueEG$ are $\step_\env$ and $\step_\guest$. Both guarantee
$\queueEG_1' \sim \queueEG_2'$ -- the former due to assumption~(1) of
Figure~\ref{fig:env}, and the latter due to assumption~(2) of
Figure~\ref{fig:guest}.

\underline{$\invupds(\queueGHu_1, \sH_1, \queueGHu_2, \sH_2)$}: This invariant is
affected only by transitions that change $\queueGHu$ or $\sH$ or
both. There are two such transitions: $\step_\guest$ and
$\step_\hypace$.

$\step_\guest$ adds to $\queueGHu$. First, note that due to the clause
$\queueGHu_1'' \sim \queueGHu_2''$ in assumption (2) of
Figure~\ref{fig:guest}, the updates $\queueGHu_1''$ and
$\queueGHu_2''$ that are added to $\queueGHu_1$ and $\queueGHu_2$
agree in their content and have the same effective times (they may
differ in when they reach HyPace, but this is irrelevant for
$\invupds(\queueGHu_1, \sH_1, \queueGHu_2, \sH_2)$ as it projects each
$\updates$ to $\updates\res{\metaonly}$). Second, due to invariant
$(I2)$, all events $\queueGHu_1$ have effective times less than
$\teffmax_1$, while the new events being added have timestamps greater
than $\teffmax_1$ due to assumption~(5) of Figure~\ref{fig:guest}. It
follows that $\queueGHu_1''$ is simply \emph{appended} to the end of
$\queueGHu_1$ by the $\merge$ operation in rule ($\rname{guest}$) of
Figure~\ref{fig:guest}. A similar observation holds for
$\queueGHu_2''$ and $\queueGHu_2$. This immediately implies that
$\invupds(\queueGHu_1', \sH_1', \queueGHu_2', \sH_2')$ holds.

$\step_\hypace$ modifies $\queueGHu$ by taking a prefix of it and
applying it to profiles. Hence, it trivially yields
$\invupds(\queueGHu_1', \sH_1', \queueGHu_2', \sH_2')$.

\underline{$\queueHE_1 \sim \queueHE_2$}: The two transitions that modify
$\queueHE$ are $\step_\env$ and $\step_\hypace$. Of these,
$\step_\env$ guarantees $\queueHE_1' \sim \queueHE_2'$ due to
assumption~(1) of Figure~\ref{fig:env}.

Showing that $\step_\hypace$ guarantees $\queueHE_1' \sim \queueHE_2'$
is harder. First, note that from the rule (\rname{hypace}) of
Figure~\ref{fig:hypace}, $\queueHE_1' = \queueHE_1 \cup \queueHE_1''$
and, similarly, $\queueHE_2' = \queueHE_2 \cup
\queueHE_2''$. $\queueHE_1 \sim \queueHE_2$ by assumption about the
invariant holding before the transition, so we only need to prove that
$\queueHE_1'' \sim \queueHE_2''$. Now, again following the rule,
$\queueHE_i''$ (for $i = 1,2$) is obtained from the function $\fHo$,
so by assumption~(7) of Figure~\ref{fig:hypace}, we only need to show
that $\resf{\sH_1'}{\tglob_1} = \resf{\sH_2'}{\tglob_2}$. We already
know from the invariant before the transition that $\tglob_1 =
\tglob_2 = \tglob$ (say) and, following the definition of
$\resf{\sH'}{\tglob}$, we only need to show that for every $i \in \{1,
\ldots, N\}$, $\resf{\prof_{i1}'}{\tglob} =
\resf{\prof_{i2}'}{\tglob}$. However, from the previous clause
($\invupds(\queueGHu_1', \sH_1', \queueGHu_2', \sH_2')$) we know that
$\prof_{i1}'$ and $\prof_{i2}'$ only differ in \emph{which} of two
identical sets of updates have been applied. However, all updates with
\emph{effective} times less than $\tglob$ \emph{must have been
  applied} to both. To see this, note that the time at which any such
update comes to HyPace (the timestamp $\tupd_i$) has to be lower than
the effective time due to invariant $(I1)$ and, hence, lower than
$\tglob$. So, $\fupdprof(\prof_i, \updates_i, \tglob)$ in the rule
($\rname{hypace}$) will apply all such updates (in both
runs). $\prof_{i1}'$ and $\prof_{i2}'$ can still differ in applied
updates with effective timestamps \emph{after} $\tglob$. However, due
to assumption~(6) of Figure~\ref{fig:hypace}, such differences are
irrelevant for the projections $\resf{\prof_{i1}'}{\tglob}$ and
$\resf{\prof_{i2}'}{\tglob}$. Hence, $\resf{\prof_{i1}'}{\tglob} =
\resf{\prof_{i2}'}{\tglob}$, as required.

\underline{$\tglob_1 = \tglob_2$}: The only transition that modifies
$\tglob$ is $\step_\hypace$, but this transition increases $\tglob$ by
a fixed amount ($\tepoch$), so it trivially preserves equality of
$\tglob_1$ and $\tglob_2$.
%
\end{proof}

\paragraph{Security}
We formulate security as follows standard noninterference. Consider
two runs that both start from empty queues, the same states for
clients and the network, the same non-private state for the guest, but
possibly \emph{different private guest state}. Then, after $n$ steps
in each run, the non-private state of the environment is exactly the
same. Let $\emptyset$ denote an empty queue.

\begin{thm}[Security]
Let (1)~$\sG_1 \sim \sG_2$, (2)~$(\sE, \sG_1, \sH, \emptyset,
\emptyset, \emptyset, \tglob) \step^n (\sE_1', \_, \_, \_, \_, \_,
\_)$ and (3)~$(\sE, \sG_2, \sH, \emptyset, \emptyset, \emptyset,
\tglob) \step^n (\sE_2', \_, \_, \_, \_, \_, \_)$. Then $\sE_1' \sim
\sE_2'$.
\end{thm}
\begin{proof}
It is trivial to see that $(I1)$, $(I2)$ and $(I3)$ hold of the
starting states. Since these properties are invariants, they hold of
the final states. In particular, from $(I3)$ on the final state, we
get that $\sE_1' \sim \sE_2'$.
\end{proof}

\end{document}